\documentclass[journal,twoside]{IEEEtran}
\usepackage{amsmath,amsthm,amsfonts,amssymb,latexsym,extarrows,booktabs,array,arydshln}
\usepackage{fancyhdr,graphicx,tabularx}
\usepackage{multirow}
\usepackage{makecell}
\usepackage{amsfonts}
\usepackage{graphicx}
\usepackage{subfigure}
\usepackage {cite}
\usepackage {url}
\usepackage {stfloats}
\usepackage {mathrsfs}
\usepackage{algorithm}
\usepackage{algorithmic}
\usepackage{xcolor}
\usepackage{longtable}
\usepackage[OT2,T1]{fontenc}

\IEEEoverridecommandlockouts
\begin{document}
 \title{\huge Code-Frequency Block Group Coding for Anti-Spoofing Pilot Authentication in Multi-Antenna OFDM Systems}
\author{Dongyang~Xu,~\IEEEmembership{Student Member,~IEEE,}
        Pinyi~Ren,~\IEEEmembership{Member,~IEEE,}
        James~A.~Ritcey,~\IEEEmembership{Fellow,~IEEE,}
        Yichen~Wang,~\IEEEmembership{Member,~IEEE}\
}
\maketitle

\begin{abstract}
A pilot spoofer can paralyze the channel estimation in  multi-user orthogonal frequency-division multiplexing  (OFDM)  systems  by using the same publicly-known  pilot tones as legitimate nodes. This causes  the problem of pilot authentication (PA). To solve this, we propose, for a two-user multi-antenna OFDM system,  a code-frequency block group (CFBG) coding based  PA  mechanism. Here  multi-user pilot information, after  being  randomized independently to avoid being spoofed,  are converted  into activation  patterns of subcarrier-block groups on code-frequency domain. Those patterns, though overlapped and interfered mutually  in the wireless transmission environment,  are qualified to be separated and identified  as the original pilots  with high accuracy, by exploiting CFBG coding theory and channel characteristic. Particularly,  we  develop the CFBG code  through  two steps, i.e., 1) devising an ordered signal detection  technique to recognize the number of  signals coexisting on each subcarrier block, and encoding each subcarrier block with the detected number; 2)  constructing  a zero-false-drop (ZFD) code and block detection based (BD) code via  $k$-dimensional Latin hypercubes and integrating those two codes into the CFBG code.  This code can bring  a desirable pilot separation error probability (SEP),  inversely proportional  to the number of  occupied subcarriers  and antennas  with a  power of $k$. To apply the code to  PA, a scheme of pilot conveying, separation and  identification  is proposed. Based on this novel PA,  a joint channel estimation and identification mechanism is proposed to  achieve high-precision channel recovery and simultaneously enhance PA  without occupying  extra resources.  Simulation results verify the effectiveness of our proposed mechanism.
\end{abstract}

\IEEEpeerreviewmaketitle
\section{Introduction}
\label{i8ntroduction}
\IEEEPARstart{S}{ecurity} in  mobile  radio communication  systems embraces a set of ideas, including authentication~\cite{Maurer}, confidentiality~\cite{Gopala,Xu_Optimal,Wu0,Xu_Artificial}, integrity, among others. Basically, authentication functions as the foremost security mechanism  since  it guarantees the identities of legitimate entities and authentic  data.  This calls for  two paradigms of authentication, including the entity authentication that very often justifies the identities of the parties taking part,  and the data origin authentication that aims to confirm the identity of a data originator~\cite{Shiu}.  With those two functionalities increasingly challenged by the rise of novel security threats, upper-layer authentication and physical layer authentication (PLA) gradually comes to serve as  two  necessary  implementation techniques  throughout the current network protocol stack. For the upper layer authentication, identity messages  are encrypted via cryptographic method whereas,  for the PLA mechanism, a legitimate terminal is  authenticated  if  its destination node can successfully demodulate and decode its transmission~\cite{Yu}. In reality,  PLA, as a complementary mechanism,  helps enhance the overall authentication efficiency.

In spite of such a comprehensive  authentication architecture,  security issues, rather than  suffering a weakening trend,  have been increasingly raised when  the network becomes more complicated~\cite{Xu_Weighted} and the threats grow  more powerful~\cite{Xu_hybrid}.  Meanwhile, though we have to also admit  the vulnerability of  upper-layer authentication under intruders with  massive computing power~\cite{Mukherjee}, some physical-layer protocols, e.g., those verifying communicating identities through publicly-known PLA but without being unprotected,   actually now  arouse huge attentions from adversaries that  can easily spoof  those identities, at least  without too much overheads.  This seems a better choice for any malicious entity  and thus legitimate systems  require more specialized   mechanism to safeguard  PLA and  protect its effectiveness~\cite{Wang}.

 For example,  OFDM technique, being universally deployed  in current commercial and military applications,  is  very vulnerable to the  security breaches on its predefined  protocols. These agreements, necessarily  configured between transceiver pairs,  are originally designed to combat the multi-path influence in wireless environment~\cite{Shahriar1}. A well-known protocol in OFDM systems is to share a predefined signal structure known as the pilot tone. The signal, like the pilot symbol employed in various networks~\cite{Xu_AF,Xu_towards}, actually acts as  a key between transceiver pair for  acquiring  channel state information (CSI)~\cite{Ozdemir}.  Basically, this process is also a kind of PLA that  authenticates the sender and receiver,  since the authentication signal from a legitimate sender, that is, \emph{publicly-known} and \emph{deterministic}  pilot tone,  is verified and, therefore, known at the receiver. However, an adversary that is aware of the specific pilot tones  used,   can nowadays exploit this to spoof the network. This is done, in theory,  by perfectly imitating the  pilot tones of a legitimate terminal  instead of aggravating data payload directly~\cite{Xu,Xu_ICA_Based}. This spoofing behavior can completely break down the  uniqueness of the traditional pilot-sharing protocols and induce contaminated and imprecise channel  estimation samples that are  then not recovered.

This motivates us to develop  the concept of  pilot authentication (PA), kind of secure and data-origin  PLA  mechanism for wireless  OFDM systems, namely,  can the pilot tone from any legitimate node  be authenticated  through wireless multiuser  channels while  hardly being spoofed? We show that the answer is yes, with the performance subject to specifically identified tradeoffs between the time-frequency-domain resources  and antenna resources.  The scenario we consider is  an uplink  multi-antenna OFDM system where  two legitimate  users, respectively named  as Bob and Charlie,   communicate with an uplink receiver Alice  threatened by a spoofer denoted by Eva. Unlike the anti-spoofing  mechanism in~\cite{Xu,Xu_ICA_Based} for single-user protection, one more user incurs a significant difficulty  on  countermeasures. The key challenge lies in the fact that Alice has to avoid the attack and  simultaneously guarantee  the PA between legitimate nodes, i.e., Bob  and  Charlie.

Therefore, we, in this paper,  first address the design issue of PA that could resolve above challenges. The first step we introduce is to \emph{randomize} the pilot tones. The randomization incurs  a  hybrid attack that embraces  spoofing, silence and  jamming behaviors but inspires us  to rethink and redesign   the fundamental PA process through  three key procedures, i. e.,  pilot conveying, separation and  identification. A code-frequency block group (CFBG) coding  based  PA mechanism is proposed  in which subcarrier blocks  are encoded to authenticate  pilots  and simultaneously reused  for   channel estimation. This mechanism  reuses the time-frequency and antenna resources original for  channel estimation and therefore  requires no extra  resource support.   The related  contributions are summarized  as follows:
\begin{enumerate}
 \item Recognizing  a hybrid attack, we build up a 4-hypothesis  testing  and devise an  ordered  eigenvalue-ratio  detection technique to recognize the number of signals coexisting.  An analytical requirement of subcarriers and antennas  is derived and configured  for one  subcarrier block in such a way  that precise number of signals   can be identified on the block. The number  is  encoded into binary number  information and,  therefore, each subcarrier  block can be precisely encoded with binary number  information.
 \item Thanks to the coded subcarrier  blocks, a code-frequency domain can be identified. On this domain,  we develop  a CFBG coding theory, constituted by a zero-false-drop (ZFD) code and  block detection (BD) code. To construct the ZFD code,  we  exploit the concept of  $k$-dimensional Latin hypercubes of order $q$. We validate that this code can be constructed  when  $k\ge 3$ and $q \ge 3$ and  cannot be otherwise. Interestingly, the required  BD code has the same codeword set but different codeword arithmetic principle as the ZFD code.
 \item Based on the CFBG code, we  derive a CFBG codebook through which  multiuser pilot information  is enabled  to  be conveyed,  separated and identified in the form of  codewords. This is done in practice by a proposed  block detection based codeword decoding (BDCD) algorithm. Theoretically, the  concept of  separation error probability (SEP) is formulated and proved to be  proportional to the parameter ${\left( {{1 \mathord{\left/{\vphantom {1 {{N_{{\rm{Total}}}}{N_{\rm{T}}}}}} \right. \kern-\nulldelimiterspace} {{N_{{\rm{Total}}}}{N_{\rm{T}}}}}} \right)^k}$ where ${N_{{\rm{Total}}}}$ and $N_{\rm T}$ respectively represent the number of subcarriers and antennas occupied.  Moreover, we show  how  the pilot identification error  occurs and how  the identification enhancement benefits from the previous process of pilot conveying and separation.
 \item  In order to enhance  identification and  further guarantee the channel estimation  high precision,  we develop  a joint channel estimation  and  identification mechanism. Here, a minimum-mean square error (MMSE) semi-blind estimator is devised  to  estimate the frequency-domain subcarriers (FS) and channel impulse response (CIR) of Bob and Charlie. Thanks to the estimated channels,  the diversity of spatial correlation of different nodes  is exploited, thus improving the pilot identification  efficiently. We formulate the  identification error probability (IEP) and  derive its  asymptotic expression under  a large number of antennas. Numerical results show that the non-zero IEP  occurs only when Eva has identical spatial correlation matrix with Bob and/or Charlie.
\end{enumerate}
The rest of the paper is summarized as follows.  We begin by briefly reviewing related work in Section~\ref{RW}. In Section~\ref{PM}, we present an overview of pilot spoofing attack on two-user multi-antenna OFDM systems. A framework of  CFBG coding based PA  is proposed in Section~\ref{J-PACFD}.  In what follows, four key techniques are introduced.  An attack detection method and its simulated performance  are demonstrated in Section~\ref{BD}. A code construction scheme and the codebook performance evaluation  are formulated in Section~\ref{CC}. A pilot encoding and decoding mechanism is presented in  Section~\ref{PED} and a joint channel estimation and  identification scheme is given in Section~\ref{JCEI} with comprehensive  simulation  validation. Finally, we conclude our work in Section~\ref{Conclusions}.

Notations: Boldface is used for matrixes ${\bf{A}}$. ${{\bf{A}}^*}$, ${{\bf{A}}^{\rm{T}}}$, ${{\bf{A}}^{{\rm H}}}$, ${{\bf{A}}^{{\rm +}}}$respectively denotes  conjugate,  transpose,  conjugate transpose and pseudoinverse of matrix ${\bf{A}}$.  $\left\| {\cdot} \right\|$ denotes the Euclidean norm of a vector or a matrix. ${\mathbb{E}}\left\{  \cdot  \right\}$ is the expectation operator. The operator $\otimes$ is  the Kronecker   product. ${\rm{diag}}\left\{ {\cdot} \right\}$ stands for the diagonal matrix with the column vector on its diagonal. ${\left\{ \cdot \right\}^{\rm{ + }}}$ denotes the the Moore-Penrose pseudoinverse.
\newtheorem{assumption}{Assumption}
\newtheorem{fact}{Fact}
\newtheorem{remark}{Remark}
 \newtheorem{theorem}{Theorem}
 \newtheorem{proposition}{Proposition}
\newtheorem{lemma}{Lemma}
\newtheorem{property}{Property}
\section{Related  works}
\label{RW}
Basically, PA, a kind of  data origin authentication, involves two aspects, i.e., verifying  data integrity and  authenticity. Authenticating pilot signals  under  pilot spoofing attack mainly refers to  confirming  their authenticity.  This process includes how to detect the alteration to authenticity  and how to protect and further maintain  the high authenticity.  Much work have been extensively investigated  on those areas  from narrow-band single-carrier system~\cite{Zhou,Kapetanovic1,Tugnait1,Xiong1,Xiong2,Kapetanovic2,Wu,Tugnait2}  to wide-band multi-carrier system~\cite{Clancy,Litchman,Shahriar2,Xu}.

Authors in~\cite{Zhou} introduced   for a narrow-band single-carrier system a pilot spoofing attack, that is,  an active eavesdropper  disturbs the normal channel estimation  by transmitting the same pilot signals as the legitimate nodes.  Following~\cite{Zhou}, much  research has studied the spoofing detection by exploiting the physical layer information, such as auxiliary  training or data sequences~\cite{Kapetanovic1,Tugnait1,Xiong1,Xiong2} and  some prior-known channel information~\cite{Kapetanovic2,Wu}. Different from those detection oriented schemes,  the author in~\cite{Tugnait2}  proposed  a   joint spoofing detection and mitigation strategy to protect the  authenticity of channel estimation samples.  When  a spoofing attack is detected,   the contaminated part  of the pilot-superimposed data   is deleted and  then   the remnant  data part is employed to achieve authentication and  estimate  CSI.

  The attack methodology on OFDM systems becomes very different   since  an intelligent spoofer, actually serving as a protocol-aware attacker, can stealthily imitate any behaviors of legitimate nodes  except  a completely random behavior~\cite{Litchman}.   Therefore, the common sense of  countermeasures is to completely  randomize the  locations and values of regular pilot tones.  Clancy et al. in~\cite{Clancy} first introduced the behavior  of  misguiding the CSI estimation process by spoofing pilot tones in OFDM systems. Employing  randomized pilot  tones with their locations obeying different probability distributions,  authors in~\cite{Shahriar2}  presented  a comprehensive analysis of decoding benefits brought by pilot randomization.  Besides those, authors in~\cite{Xu} proposed a pilot encoding-and-decoding  mechanism to achieve robust PA while providing precise  CSI estimation.

  However, those work only focus on the single-user scenario and do not specify the PA issue existing in practical  multi-user OFDM systems  over frequency-selective fading channels.
\section{Pilot Spoofing Attack on Two-User Multi-Antenna OFDM Systems:\\ Overview and Challenges}
\label{PM}
We in this section  begin our discussion by outlining a fundamental  overview of pilot spoofing attack, including  the basic system  and  problem model as well as the signal and  channel estimation model. Then we describe  a common-sense  technique, i.e.,  pilot randomization,  to defend against  pilot spoofing attack and identify  the existing key challenges.
\subsection{System Description and Problem Model}
 We consider an uplink two-user single-input multiple-output (SIMO)-OFDM  systems where an uplink receiver  named Alice is equipped with  $N_{\rm T}$ antennas  and  two uplink transmitters, respectively denoted by Bob and Charlie, are each configured with  single antenna.  A block diagram of such a system using time division duplex (TDD) mode  over  frequency-selective fading channels is depicted in Fig.~\ref{System_model}. Pilot tone based channel estimation is considered in the uplink.  Conventionally,  PA, an unavoidable step before channel estimation,  is achieved by assigning Bob and Charlie  with   \emph{publicly-known and deterministic} pilot tones that can be identified. This  mechanism, a  kind of data-driven PLA,  is very  fragile and actually has no privacy.  The problem is that a malicious node  Eva  with single antenna can impersonate Bob or Charlie synchronously by using the same pilot tones, without need of imitating their identities. In this way,  Eva can misguide the multi-user channel estimation that is acquired at Alice by linear decorrelation based on pilot tones. The disturbed CSI, once utilized for downlink transmission in TDD systems, can induce serious information leakage to Eva.
\subsection{SIMO Received Signal Model}
\begin{figure}[!t]
\centering\includegraphics[width=1.0\linewidth]{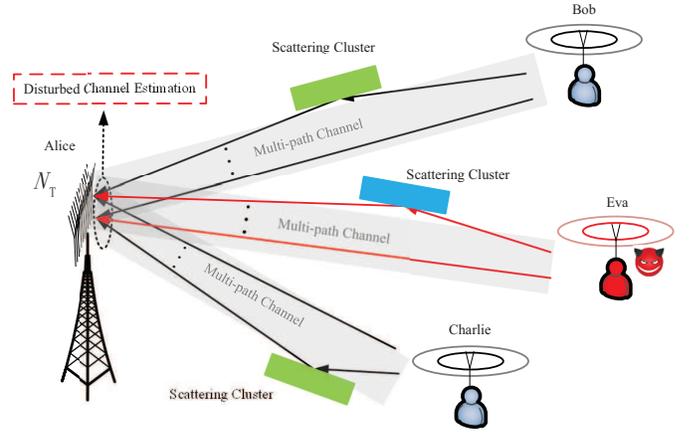}
\caption{Uplink two-user OFDM system model  under pilot spoofing attack.}
\label{System_model}
\end{figure}
Let us first turn to the representation of  signal model. At each transmit (receive) antenna of nodes, the conventional OFDM modulator (demodulator) is equipped to  map bit streams into frequency-domain signals transmitted on $N$ subcarriers. OFDM symbols transmitted from Bob, Charlie and Eva at time index $k$ are denoted by vectors ${{\bf{x}}_{{j}}}\left[ k \right]\in {{\mathbb C}^{N\times 1}},  j \in \left\{ {{\rm{B,C,E}}} \right\}$. Those vectors are processed by inverse fast Fourier transform (IFFT) and  then each added with a cyclic prefix of length $p$ to combat the multi-path influence. Generally, it is assumed that $ p \ge L - 1$, where $L$ is the maximum length of all channels. After removing the cyclic prefix at the $i$-th receive antenna, Alice derives the time-domain signal vector ${{\bf{y}}^i}\left[ k \right]\in {{\mathbb C}^{N \times 1}}$ written as
\begin{eqnarray}\label{E.1}
 \hspace{-10pt}
{{\bf{y}}^i}\left[ k \right]  \hspace{-5pt}&=&\hspace{-5pt} {\bf{H}}_{{\rm{C,B}}}^i{{\bf{F}}^{\rm{H}}}{{\bf{x}}_{\rm{B}}}\left[ k \right] + {\bf{H}}_{{\rm{C,C}}}^i{{\bf{F}}^{\rm{H}}}{{\bf{x}}_{\rm{C}}}\left[ k \right] \nonumber \\
 \hspace{-5pt}&+&\hspace{-5pt}{\bf{H}}_{{\rm{C,E}}}^i{{\bf{F}}^{\rm{H}}}{{\bf{x}}_{\rm{E}}}\left[ k \right] + {{\bf{v}}^i}\left[ k \right]
\end{eqnarray}
where ${\bf{H}}_{{\rm{C,B}}}^i$, ${\bf{H}}_{{\rm{C,C}}}^i$ and ${\bf{H}}_{{\rm{C,E}}}^i$ are the $N \times N$ circulant matrices of Bob, Charlie and Eva, with the  first column respectively given by ${\left[ {\begin{array}{*{20}{c}}
{{\bf{h}}_{\rm{B}}^{{i^{\rm{T}}}}}&{{{\bf{0}}_{1 \times \left( {N - L} \right)}}}
\end{array}} \right]^{\rm{T}}}$, ${\left[ {\begin{array}{*{20}{c}}
{{\bf{h}}_{\rm{C}}^{{i^{\rm{T}}}}}&{{{\bf{0}}_{1 \times \left( {N - L} \right)}}}
\end{array}} \right]^{\rm{T}}}$,  and ${\left[ {\begin{array}{*{20}{c}}
{{\bf{h}}_{\rm{E}}^{{i^{\rm{T}}}}}&{{{\bf{0}}_{1 \times \left( {N - L} \right)}}}
\end{array}} \right]^{\rm{T}}}$.
Those $L$ by $1$ CIR vectors to the $i$-th receive antenna of Alice, i.e., ${\bf{h}}_{\rm{B}}^i $, ${\bf{h}}_{\rm{C}}^i$  and ${\bf{h}}_{\rm{E}}^i$, are  mutually independent with each other. The channel power delay profile (PDP) of Bob, Charlie and Eva at  the $j$-th path to the $i$-th antenna of Alice are respectively  denoted  by  $\sigma _{{\rm{B}},j,i}^2$,  $\sigma _{{\rm{C}},j,i}^2$, and  $\sigma _{{\rm{E}},j,i}^2$.  Without loss of generality, channel PDPs   are normalized  so that  $\sum\limits_{j = 1}^L {\sigma _{{l},j,i}^2}  = 1, \forall i, \forall l \in \left\{ {{\rm{B,C,E}}} \right\}$ are satisfied.  The CIRs of different paths   exhibit spatially uncorrelated Rayleigh fading for each receiving antenna and CIRs  of   different antennas are assumed to be spatially correlated  for each path. The receive  correlation matrix of signals from  Bob, Charlie and Eva are respectively  denoted by ${\bf R}_{\rm B}$, ${\bf R}_{\rm C}$ and ${\bf R}_{\rm E}$.  ${{\bf{v}}^i}\left[ k \right]\in {{\mathbb C}^{N \times 1}}$ denotes the vector of i.i.d. random variable satisfying ${\cal C}{\cal N}\left( {0,{\sigma ^2}} \right)$ where $\sigma ^2$ is noise power.  ${\bf{F}}$ denotes the $N\times N$ unitary DFT matrix and  it is easy to show that the eigenvalue decomposition of ${\bf{H}}_{{\rm{C}},j}^i,  j \in \left\{ {{\rm{B,C,E}}} \right\}$ leads to ${\bf{H}}_{{\rm{C}},j}^i = {{\bf{F}}^{\rm{H}}}{\rm{diag}}\left\{ {\sqrt N {\bf{F}}{{\left[ {\begin{array}{*{20}{c}}
{{\bf{h}}_{{j}}^{{i^{\rm{T}}}}}&{{{\bf{0}}_{1 \times \left( {N - L} \right)}}}
\end{array}} \right]}^{\rm{T}}}} \right\}{\bf{F}},  j \in \left\{ {{\rm{B,C,E}}} \right\}$. Taking  FFT of received signals,  Alice  finally obtains the version of $N$ by $1$  frequency-domain signals ${\widetilde {\bf{y}}^i}\left[ k \right] $ at the $i$-th receive antenna as
\begin{eqnarray}\label{E.2}
 \hspace{-5pt}
{\widetilde {\bf{y}}^i}\left[ k \right]   \hspace{-5pt}&=&\hspace{-5pt} {\rm{diag}}\left\{ {\sqrt N {\bf{F}}{{\left[ {\begin{array}{*{20}{c}}
{{\bf{h}}_{\rm{B}}^{{i^{\rm{T}}}}}&{{{\bf{0}}_{1 \times \left( {N - L} \right)}}}
\end{array}} \right]}^{\rm{T}}}} \right\}{{\bf{x}}_{\rm{B}}}\left[ k \right]\nonumber \\
  \hspace{-5pt}&+&\hspace{-5pt}
{\rm{diag}}\left\{ {\sqrt N {\bf{F}}{{\left[ {\begin{array}{*{20}{c}}
{{\bf{h}}_{\rm{C}}^{{i^{\rm{T}}}}}&{{{\bf{0}}_{1 \times \left( {N - L} \right)}}}
\end{array}} \right]}^{\rm{T}}}} \right\}{{\bf{x}}_{\rm{C}}}\left[ k \right] \nonumber \\
  \hspace{-5pt}&+&\hspace{-5pt}{\rm{diag}}\left\{ {\sqrt N {\bf{F}}{{\left[ {\begin{array}{*{20}{c}}
{{\bf{h}}_{\rm{E}}^{{i^{\rm{T}}}}}&{{{\bf{0}}_{1 \times \left( {N - L} \right)}}}
\end{array}} \right]}^{\rm{T}}}} \right\}{{\bf{x}}_{\rm{E}}}\left[ k \right] + {{\bf{w}}^i}\left[ k \right]\nonumber \\
\end{eqnarray}
where ${{\bf{w}}^i}\left[ k \right] = {\bf{F}}{{\bf{v}}^i}\left[ k \right]\in {{\mathbb C}^{N \times 1}}$ is the DFT projection of the random vector ${{\bf{v}}^i}\left[ k \right]$.  We see that
since  ${{\bf{v}}^i}\left[ k \right]$ is isotropic, ${{\bf{w}}^i}\left[ k \right]$ has the same distribution as  ${{\bf{v}}^i}\left[ k \right]$, i.e., a vector of i.i.d.  ${\cal C}{\cal N}\left( {0,{\sigma ^2}} \right)$ random variables.
After  simplification,  the received signal  is transformed into:
\begin{eqnarray}\label{E.3}
 \hspace{-10pt}
{\widetilde {\bf{y}}^i}\left[ k \right]   \hspace{-5pt}&=&\hspace{-5pt} {\rm{diag}}\left\{ {{{\bf{x}}_{\rm{B}}}\left[ k \right]} \right\}{{\bf{F}}_{\rm{L}}}{\bf{h}}_{\rm{B}}^{{i}} +
{\rm{diag}}\left\{ {{{\bf{x}}_{\rm{C}}}\left[ k \right]} \right\}{{\bf{F}}_{\rm{L}}}{\bf{h}}_{\rm{C}}^{{i}}\nonumber \\
 \hspace{-5pt}&+&\hspace{-5pt}{\rm{diag}}\left\{ {{{\bf{x}}_{\rm{E}}}\left[ k \right]} \right\}{{\bf{F}}_{\rm{L}}}{\bf{h}}_{\rm{E}}^{{i}} + {{\bf{w}}^i}\left[ k \right]
\end{eqnarray}
where ${{\bf{F}}_{\rm{L}}} = \sqrt N {\bf{F}}\left( {:,1:L} \right)$.  Next, we make  assumptions:
\begin{assumption}
 We assume ${{\bf{x}}_{\rm{B}}}\left[ k \right] = {x_{\rm{B}}}\left[ k \right]{{\bf{1}}_{N \times 1}}$ and ${{\bf{x}}_{\rm{C}}}\left[ k \right] = {x_{\rm{C}}}\left[ k \right]{{\bf{1}}_{N \times 1}}$ where ${{\bf{1}}_{N \times 1}}$ is a column vector whose elements are equally to be one.  Alternatively, we can superimpose those pilots onto a dedicated pilot sequence optimized under a non-security oriented scenario and utilize this new pilot for training. At this point,  ${{\bf{x}}_{\rm{B}}}\left[ k \right]$,  ${{\bf{x}}_{\rm{C}}}\left[ k \right]$ can be an additional phase difference for security consideration. Signals transmitted by Eva can  be  denoted by ${\rm{diag}}\left\{ {{{\bf{x}}_{\rm{E}}}\left[ k \right]} \right\} = {x_{\rm{E}}}\left[ k \right]{\bf{E}}$ where  ${\bf{E}}$ has dimension $N\times N$ and is unknown.
\end{assumption}

We denote the  pilot tones at $k$-th symbol time by ${x_{\rm{B}}}\left[ k \right],{x_{\rm{C}}}\left[ k \right],{x_{\rm{E}}}\left[ k \right]$ with ${x_{\rm{B}}}\left[ k \right] = \sqrt {{\rho _{\rm{B}}}} {e^{j{\theta _k}}},{x_{\rm{C}}}\left[ k \right] = \sqrt {{\rho _{\rm{C}}}} {e^{j{\beta _k}}}, {x_{\rm{E}}}\left[ k \right] = \sqrt {{\rho _{\rm{E}}}} {e^{j{\varphi _k}}}$ where $\rho _{\rm{B}}$, $\rho _{\rm{C}}$ and $\rho _{\rm{E}}$  respectively denote the transmitting  power of Bob, Charlie and Eva.
\begin{assumption}
 We assume that pilot tones across adjacent symbol time   are kept with fixed phase difference for each legitimate node. In this principle, we define ${\theta _{j+1}} - {\theta _j}  = \theta , j\ge0 $,  ${\beta _{j+1}} - {\beta _{j}}= \beta,  j\ge0 $ where  ${\theta}$,  ${\beta}$  are fixed and  known by all parties.
\end{assumption}
\begin{assumption}
Alice can acquire ${\bf R}_{\rm B}$ and ${\bf R}_{\rm C}$ perfectly except  ${\bf R}_{\rm E}$.  As a basic system configuration, at least four  OFDM symbols  are assumed to be within one coherence time.
\end{assumption}
\subsection{Channel Estimation Model Under Spoofed Pilots}
Now let us turn to describe the estimation models  of  FS channels. We note that  Eq.~(\ref{E.3}) is transformed into:
\begin{equation}\label{E.4}
{\widetilde {\bf{y}}^i}\left[ k \right] = {{\bf{F}}_{\rm{L}}}{\bf{h}}_{\rm{B}}^{{i}}{x_{\rm{B}}}\left[ k \right] +
{{\bf{F}}_{\rm{L}}}{\bf{h}}_{\rm{C}}^{{i}}{x_{\rm{C}}}\left[ k \right]+
{\bf{E}}{{\bf{F}}_{\rm{L}}}{\bf{h}}_{\rm{E}}^{{i}}{x_{\rm{E}}}\left[ k \right] + {{\bf{w}}^i}\left[ k \right]
\end{equation}
The spoofing  pilots  make ${\bf{E}}$ an identity matrix.  Stacking the received signals across  $N_{\rm T}$ antennas, Alice obtains the received signals as:
\begin{equation}\label{E.5}
{\bf{y}}\left[ k \right] = {x_{\rm{B}}}\left[ k \right]{{\bf{h}}_{\rm{B}}} + {x_{\rm{C}}}\left[ k \right]{{\bf{h}}_{\rm{C}}} +{x_{\rm{E}}}\left[ k \right]{{\bf{h}}_{\rm{E}}} + {\bf{w}}\left[ k \right]
\end{equation}
Here, we have  ${\bf{y}}\left[ k \right] = \left[ {\begin{array}{*{20}{c}}
{{\widetilde {{\bf{y}}}^{{1^{\rm{T}}}}}\left[ k \right]}&{, \ldots ,}&{{\widetilde {{\bf{y}}}^{{N_{\rm{T}}}^{\rm{T}}}}\left[ k \right]}
\end{array}} \right]\in {{\mathbb C}^{1 \times NN_{\rm T}}}$, ${{\bf{h}}_{{j}}} = \left[ {\begin{array}{*{20}{c}}
{{{\left( {{{\bf{F}}_{\rm{L}}}{\bf{h}}_{{j}}^1} \right)}^{\rm{T}}}}&{, \ldots ,}&{{{\left( {{{\bf{F}}_{\rm{L}}}{\bf{h}}_{{j}}^{{N_{\rm{T}}}}} \right)}^{\rm{T}}}}
\end{array}} \right] \in {{\mathbb C}^{1 \times NN_{\rm T}}}$, $j \in \left\{ {{\rm{B,C,E}}} \right\}$   and there exists ${\bf{w}}\left[ k \right] = \left[ {\begin{array}{*{20}{c}}
{{{\bf{w}}^{{1^{\rm{T}}}}}\left[ k \right]}&{, \ldots ,}&{{{\bf{w}}^{{N_{\rm{T}}}^{\rm{T}}}}\left[ k \right]}
\end{array}} \right] \in {{\mathbb C}^{1 \times NN_{\rm T}}}$.  Collecting signals within two OFDM symbols, i.e. ${k_1}$ and ${k_2}$,  Alice can further derive
\begin{equation}\label{E.6}
{\widehat {\bf{Y}}}{\rm{ = }}{{\bf{x}}_{\rm{B}}}{{\bf{h}}_{\rm{B}}}{\rm{ + }}{{\bf{x}}_{\rm{C}}}{{\bf{h}}_{\rm{C}}} + {{\bf{x}}_{\rm{E}}}{{\bf{h}}_{\rm{E}}} + {\bf{w}}
\end{equation}
where ${\widehat {\bf{Y}}} = {\left[ {\begin{array}{*{20}{c}}
{{{\bf{y}}^{\rm{T}}}\left[ {{k_1}} \right]}&{{{\bf{y}}^{\rm{T}}}\left[ {{k_2}} \right]}
\end{array}} \right]^{\rm{T}}}$, ${{\bf{x}}_{{j}}} = {\left[ {\begin{array}{*{20}{c}}
{{x_{{j}}}\left[ {{k_1}} \right]}&{{x_{{j}}}\left[ {{k_2}} \right]}
\end{array}} \right]^{\rm{T}}}$, $ j \in \left\{ {{\rm{B,C,E}}} \right\}$  and there exists ${\bf{w}} = {\left[ {\begin{array}{*{20}{c}}
{{{\bf{w}}^{\rm{T}}}\left[ {{k_1}} \right]}&{{{\bf{w}}^{\rm{T}}}\left[ {{k_2}} \right]}
\end{array}} \right]^{\rm{T}}}$.

We consider the configuration of orthogonal pilots, namely,  ${\bf{x}}_{\rm{B}}^ + {{\bf{x}}_{\rm{C}}} = 0,{\bf{x}}_{\rm{C}}^ + {{\bf{x}}_{\rm{B}}} = 0$,  to  put an explicit interpretation on  the security problem,  that is, a least square (LS) estimation of ${{{\bf{h}}_{\rm{B}}}}$ or  ${{{\bf{h}}_{\rm{C}}}}$, contaminated by  ${{{\bf{h}}_{\rm{E}}}}$  with  a noise bias, is given by:
\begin{equation}\label{E.7}
{\widehat {\bf{h}}_{{\rm{con}}}}=\left\{ {\begin{array}{*{20}{c}}
{{{\bf{h}}_{\rm{B}}} + {{\bf{h}}_{\rm{E}}} + {\bf{x}}_{\rm{B}}^ + {\bf{w}}}&{if\,\,{{\bf{x}}_{\rm{E}}} = {{\bf{x}}_{\rm{B}}}}\\
{{{\bf{h}}_{\rm{C}}} + {{\bf{h}}_{\rm{E}}} + {\bf{x}}_{\rm{C}}^ + {\bf{w}}}&{if\,\,{{\bf{x}}_{\rm{E}}} = {{\bf{x}}_{\rm{C}}}}
\end{array}} \right.
\end{equation}
Basically, employing  any nonorthogonal pilots causes the similar  phenomenon.  Therefore, the estimate value  depends on which  pilot  is spoofed by Eva.
\newtheorem{problem}{Problem}
\begin{problem}
Alice  cannot distinguish which legitimate user is being spoofed since  any  prior information of  the attack decision made by Eva is unavailable at Alice.
\end{problem}
Obviously, only one spoofer can  completely paralyze the whole channel estimation process for  multiple users.
\begin{figure}[!t]
\centering \includegraphics[width=1\linewidth]{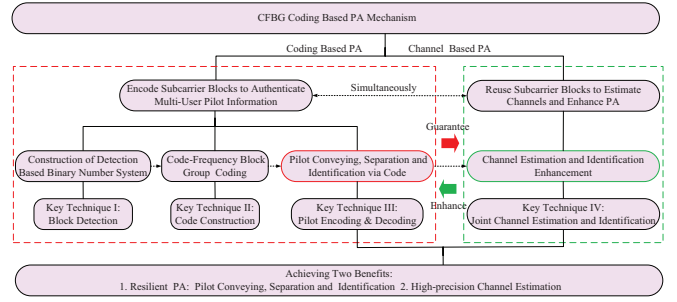}
\caption{Methodology for CFBG coding based PA,  including coding based PA and channel based PA. Particularly,  coding based PA  provides the basis of channel based PA which enhances coding based PA as well. The overall process embraces two key ideas and four vital implementing  techniques. The ultimate result is to achieve resilient PA and high-precision channel estimation simultaneously.}
\label{figure-Methodology}
\end{figure}
\begin{remark}
 For single user scenario, Alice is only required to avoid pilot spoofing attack. However, in this two-user scenario, significant difference lies in the fact that Alice has to additionally guarantee the PA between legitimate nodes. Basically, Alice has to first guarantee the PA between legitimate nodes, i.e., Bob and Charlie under the circumstance of random pilots, and then avoid the Problem 1. In fact, the random pilots incur huge difficulties for PA between Bob and Charlie and this issue becomes more challenging under hybrid attack
\end{remark}
\subsection{Novel Attack Environment and One Critical Challenge}
Pilot randomization usually serves as a prerequisite for efficiently paralyzing the pilot spoofing attack.  The commonsense is that Bob and Charlie independently  randomize their own pilot tones~\cite{Shahriar1,Xu}. In practice, the randomization of pilot tone values is employed for CIR estimation.
Theoretically,  the probability of being spoofed is zero in this case.

However,   those  pilot tones of continuous values, when utilized  for PA,  have  to be  quantized into discrete values in a  limited alphabet with high resolution,  for  convenience of sharing between transceiver~\cite{Xu}. More specifically, each of candidate pilot phases for CIR estimation  is mapped into a unique quantized sample, chosen from the set ${\cal A}$  defined by ${\cal A} = \left\{ {{e^{j{\theta }}}:{\theta }\sim U\left[ {0,2\pi } \right)} \right\}$ where $U$ denotes the uniform distribution. This can be  achieved by multiplying  the traditional  OFDM pilot tones with  suitable sequences ${\delta _{\rm{B}}}\left[ k \right]$ and ${\delta _{\rm{C}}}\left[ k \right]$.

Anyway, the pilots utilized for estimation are continuous for avoiding attack while those utilized for PA must be discrete, influenced by the quantization precision. The limited-alphabet representation of pilot tones  brings PA a novel problem:
\begin{problem}
Eva  imitates to select random pilot phases ${\varphi}$  satisfying  ${\varphi _{j+1}} - {\varphi _{j}}= \varphi,  j\ge0 $  from $\cal A$ and launches a  spoofing attack, denoted  by \textbf{randomly-imitating attack}. Moreover, Eva that is inspired  to keep silent is also able to cheat Bob and Charlie to adopt random pilots, without costing any extra resource. This is denoted by a \textbf{silence cheating} mode. What's worse, Eva can also launch \textbf{pilot jamming attack} with arbitrary jamming signals, in a stealthy way under  the ``shield'' of the enjoinment of pilot randomization . Basically, Eva can launch a \textbf{ hybrid attack}, that is, combination of  randomly-imitating  attack,  silence cheating, and pilot jamming attack. Those  behaviors are generally unpredictable.
\end{problem}
 Besides this,  pilot randomization  imposes  on PA  complex interference caused by user randomness and   independence. Under this circumstance, the randomized pilot information is non-recoverable and  thus the secure delivery of pilot information  is  challenging  in the following sense:
\begin{problem}
Those randomized and independent  pilots, if utilized for authentication through multiuser channels, will be hidden in the random channel environment and  cannot be separated, let alone identified.
\end{problem}
\section{Framework of   Code-Frequency Block Group Coding Based Pilot Authentication}
\label{J-PACFD}
In this section, we identify the key points  required for the design of secure PA. We develop a  CFBG coding based PA framework in Fig.~\ref{figure-Methodology} with the following general description for its core components. .
\subsection{Core of PA  under Hybrid Attack}
Naturally, rethinking Problem 3 inspires  us to redesign the overall PA process as pilot conveying,  separation and  identification.  Correspondingly,  we have  to  answer three questions, including 1) How to correctly convey randomized pilots of any legitimate node to Alice?  2) How to  separate multiple  pilots hidden in the wireless environment with high precision? 3) how to then reliably identify those separated  pilot?
To answer the questions mentioned above,  we identify the coding based PA  in  the following way:
\begin{fact}
 Perform pilot conveying on code domain through  a codebook medium with the potential for excellent abilities of  pilot  separation and identification  under hybrid attack.
\end{fact}
\subsection{ CFBG Coding Based PA}
For this concept,  we stress that  the  advantages of information coding and channel characteristic are exploited jointly.  Subcarrier blocks are  reused by randomized pilots   for channel estimation and  simultaneously encoded   for resilient PA.  Generally, this process includes a \textbf{coding based PA} and a \textbf{channel based PA}. The  relationship between the two methods is shown in Fig.~\ref{figure-Methodology},  embracing four steps.
\subsubsection{Step I: Construction of Detection Based Binary Number System }
Basically, in order to find the desirable codebook, we need to acquire  efficient features  that are easy to encode and decode. A fact is that the activation patterns of  subcarrier blocks of given certain size can be represented by digit 1or 0, depending on  whether those subcarriers are activated or not. Hinted by this, our goal in this part is  to  determine the block size,   precisely detect  the activation  patterns of each subcarrier block, and finally encode the results into  binary digits. To achieve this, a block detection technique  is  proposed and detailed in Section~\ref{BD}. In this way,  assuming the whole subcarriers  are divided and grouped into $B$ blocks each of which has ${{{N}} \mathord{\left/{\vphantom {{{N}} B}} \right.\kern-\nulldelimiterspace} B}$ subcarriers, we can  define a set of binary code vector as ${\cal S}= \left\{ {\left. {\bf{s}} \right|{s_i} \in \left\{ {0,1} \right\},1 \le i \le {L_s}} \right\}$ where ${L_s}$ denotes the maximum length of  the  code.
\subsubsection{Step II: Code-Frequency Block Group  Coding}
In order to formulate the codebook medium required, we first construct  a  code frequency domain on the basis of the binary number system. It is constituted by a set of pairs $\left( {{\bf{c}},b} \right)$ in Fig.~\ref{figure-CFBG_OFDM_IMPLEMENT}(a), where ${\bf{c}} \in \cal S$ and $b,1\le b \le B$ is an integer which represents the subcarrier block index of appearance of the code. $B$ is the maximum number of available blocks supported. In what follows,
by grouping and scheduling  multiple binary digits on code-frequency domain, we have the potential for formulating a codebook, for example, with   the dimension of  $B$ by  $C$ if  patterns of total number $C$ are supported. The key requirement is developing a suitable coding method such that a mapping from a codeword to  the activation patterns  is formulated and  unique patterns  can be created extensively.  Further optimizing the code, we  could construct  the codebook  achieving Fact 1 by  a technique of  CFBG  coding which is detailed  in  Section~\ref{CC}.
\subsubsection{Step III: Pilot Conveying,  Separation and Identification via Code}
Based on the theoretical  codebook, we turn to the practical construction of  conveying, separation and identification  of  pilot phase information.  At this point, pilot conveying  means  encoding pilot phases into activation patterns through the codebook.  Pilot  separation  and identification functions to achieve resilient decoding of phase information from the observed patterns, disturbed under multi-user codeword interference and hybrid attack. An implementation of the overall  process in two-user OFDM systems  is indicated   in Fig.~\ref{figure-CFBG_OFDM_IMPLEMENT}(b), including  three components, i.e., a  \textbf{Block  Identity (ID) Mapper}, a \textbf{Block Creator},  a \textbf{Detector} and a \textbf{Identifer}.
 The key is  the proposed pilot encoding and decoding technique which is further detailed in  Section~\ref{PED}.
\begin{figure}[!t]
\centering \includegraphics[width=1\linewidth]{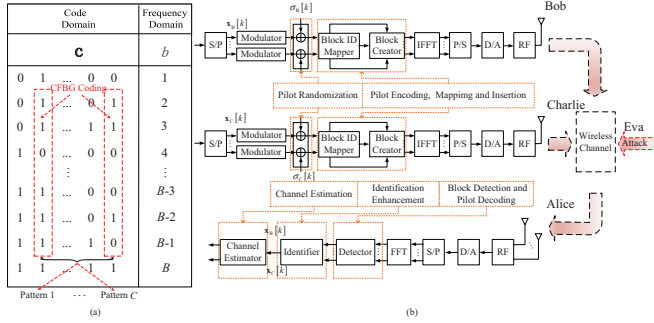}
\caption{Theoretical support and practical implementation for CFBG coding based PA mechanism; (a) General description of  CFBG coding theory on code-frequency domain; (b) Implementation framework of  CFBG coding based PA mechanism for  a two-user uplink OFDM systems. }
\label{figure-CFBG_OFDM_IMPLEMENT}
\end{figure}
\subsubsection{Step IV: Channel Estimation and Identification Enhancement}
 On the basis of coding based PA, identified  pilots are utilized for  channel estimation.   Channel based PA is performed at an \textbf{Estimator}.  The principle is that the spatial correlation property of estimated channels are employed for  enhancing  pilot identification. The detailed technique is shown in Section~\ref{JCEI}.
In the following sections, we will extend the  four key techniques in details.
\section{Key Technique I: Block Detection}
\label{BD}
In this section, we present how to exploit signal detection technique  to  execute the step I.
%
\subsection{Construction of  4-Hypotheses Testing }
Observing the possible number  $i$ of signals coexisting on one block, we can respectively  define the  hypotheses by ${{\cal H}_i}, i=0,1,2,3$ under which the  received signals stacked in  four OFDM symbols,  can be  represented by
\begin{equation}\label{E.9}
 {\bf{Y}} = {\bf{DH}} + {\bf{W}}
 \end{equation}
 Here, we have $ {\bf{Y}} = {\left[ {{{\bf{y}}^{\rm{T}}}{{\left[ {{k_i}} \right]}_{0 \le i \le 3}}} \right]^{\rm{T}}}\in {{\mathbb C}^{4 \times NN_{\rm T}}}$, and ${\bf{W}} = {\left[ {{{\bf{w}}^{\rm{T}}}{{\left[ {{k_i}} \right]}_{0 \le i \le 3}}} \right]^{\rm{T}}}\in {{\mathbb C}^{4 \times NN_{\rm T}}}$. The components of ${\bf{H}}\in {{\mathbb C}^{i \times NN_{\rm T}}}$ and ${\bf{D}}\in {{\mathbb C}^{4 \times i}}$ are selected from ${\bf{h}}_{j}$ and ${\bf{x}}_{j}$, $j \in \left\{ {{\rm{B,C,E}}} \right\}$, depending on the specific nodes coexisting on one block.
 Note that  additive vector $ {\bf{W}} $ is  independent of channel vectors $ {\bf{H}} $.  We define  the covariance matrix by ${\bf{R}} = \frac{1}{{{N_{\rm{T}}}N}}{\bf{Y}}{{\bf{Y}}^{\rm{H}}}$. According to the law of large number (LLN),  the following equation can be satisfied:
  \begin{equation}\label{E.11}
 {\bf{R}}\xlongrightarrow[{N_{\rm{T}}}N \to \infty]{ \rm{a.s.}}\frac{1}{{{N_{\rm{T}}}N}}{{\rm{\mathbb E}}_{\bf{H}}}\left\{ {{\bf{DH}}{{\bf{H}}^{\rm{H}}}{{\bf{D}}^{\rm{H}}}} \right\}{\rm{ + }}\frac{1}{{{N_T}N}}{{\rm{\mathbb E}}_{\bf{W}}}\left\{ {{\bf{W}}{{\bf{W}}^{\rm{H}}}} \right\}
  \end{equation}
  Examining the equation, we know that the rank of first term is equal to $i$ under ${{\cal H}_i}, i=0,1,2,3$ and the rank of second term is always four. Under the hypothesis of  ${{\cal H}_1}$, the   eigenvalues of ${\bf{R}}$ with the exception of   the largest one  can  all be approximately equal to the noise variance ${\sigma ^2}$. The approximation becomes exact  as ${N_{\rm{T}}}N \to \infty$. Therefore, it is possible to infer the absence or presence of the signals by comparing  the largest eigenvalue with the smallest one.  Similarly, the existence of $i$-th signals $2\le i\le3$,  under the hypothesis of  ${{\cal H}_i}$,  depends on the comparison of the $i$-th largest eigenvalue  with the smallest one.
       \begin{figure*}
 \vspace{-6pt} \centerline{
  ~\includegraphics[width=2.50in]{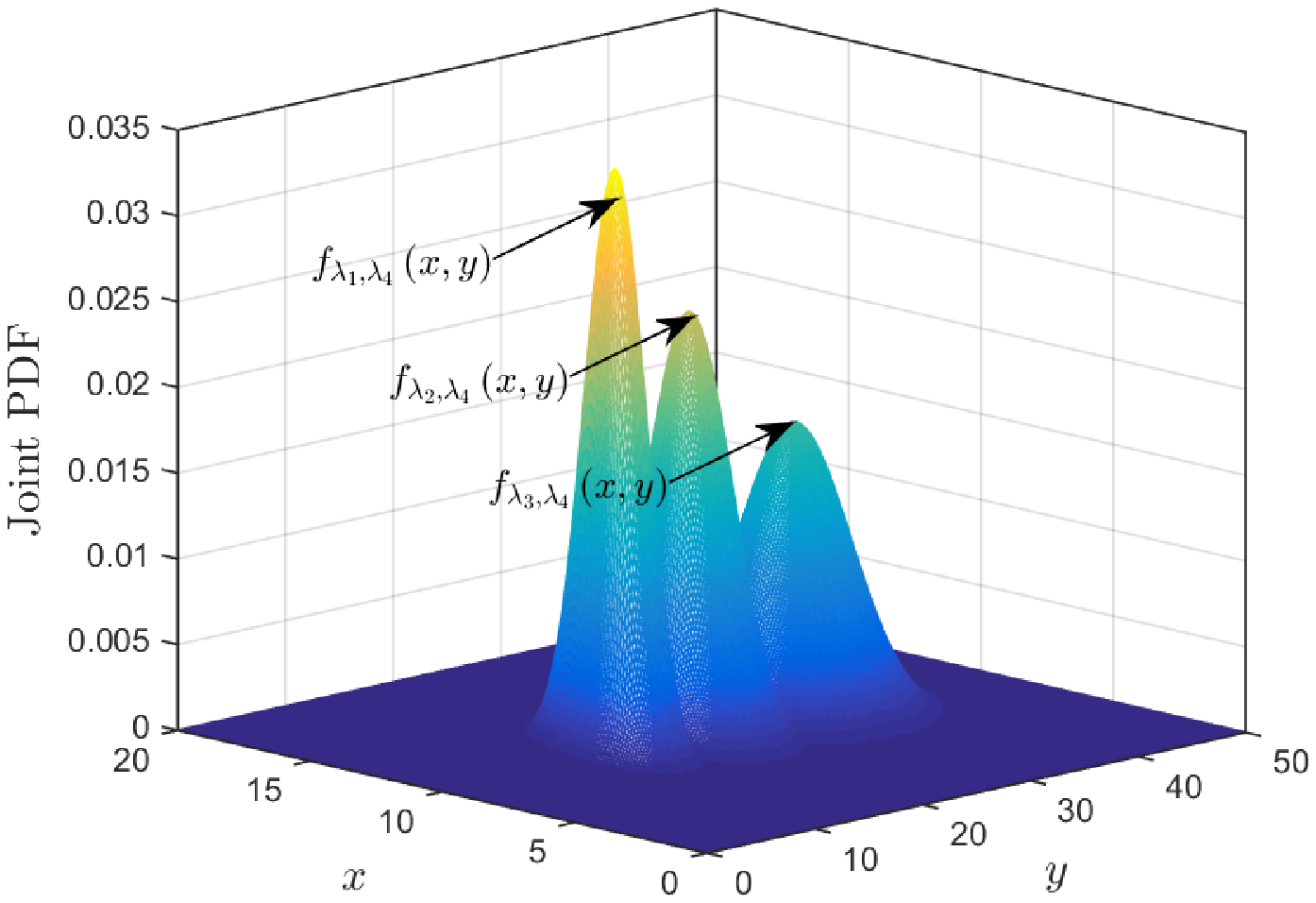} \hspace{-10pt} \includegraphics[width=2.50in]{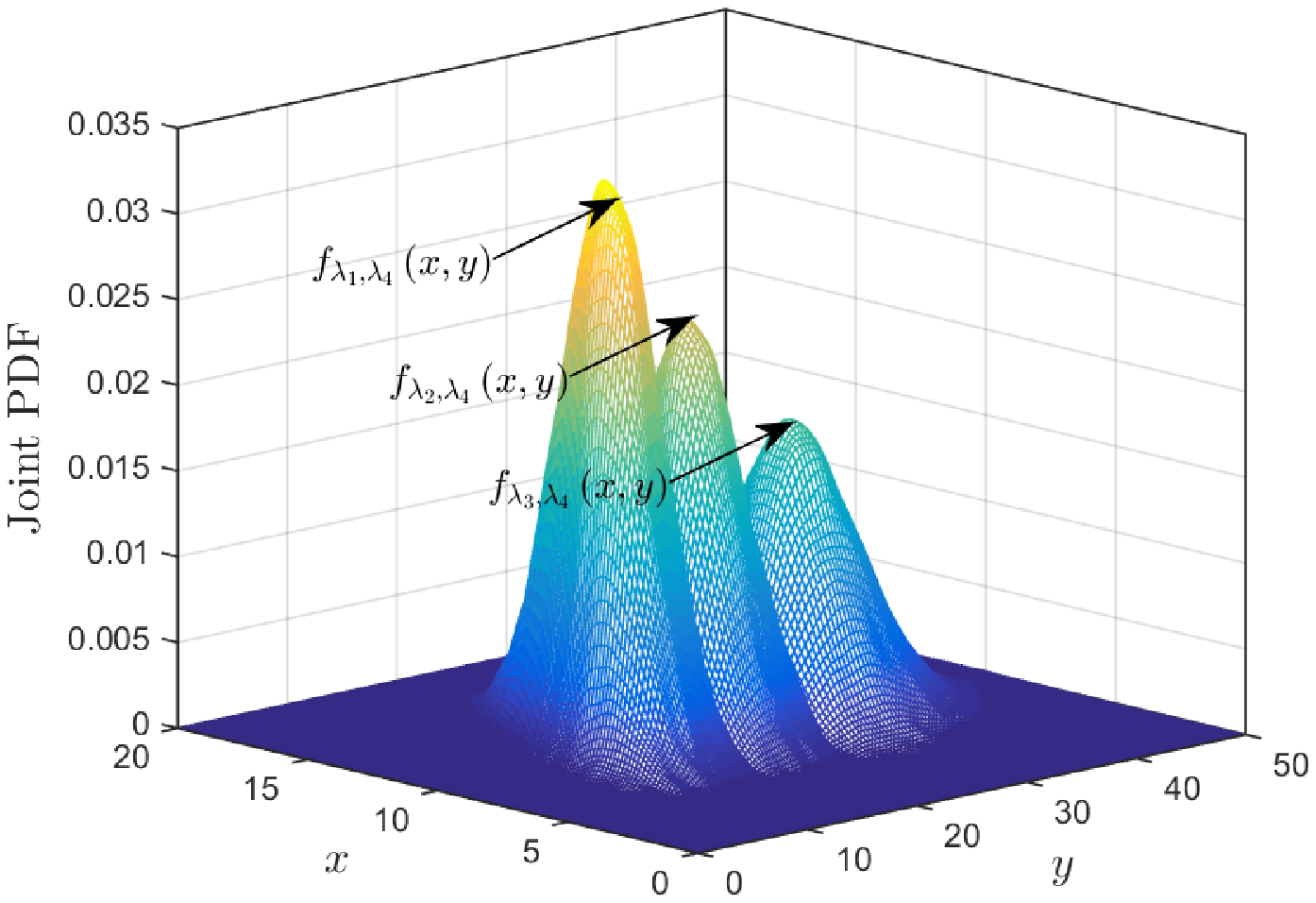}\hspace{-10pt}  \includegraphics[width=2.50in]{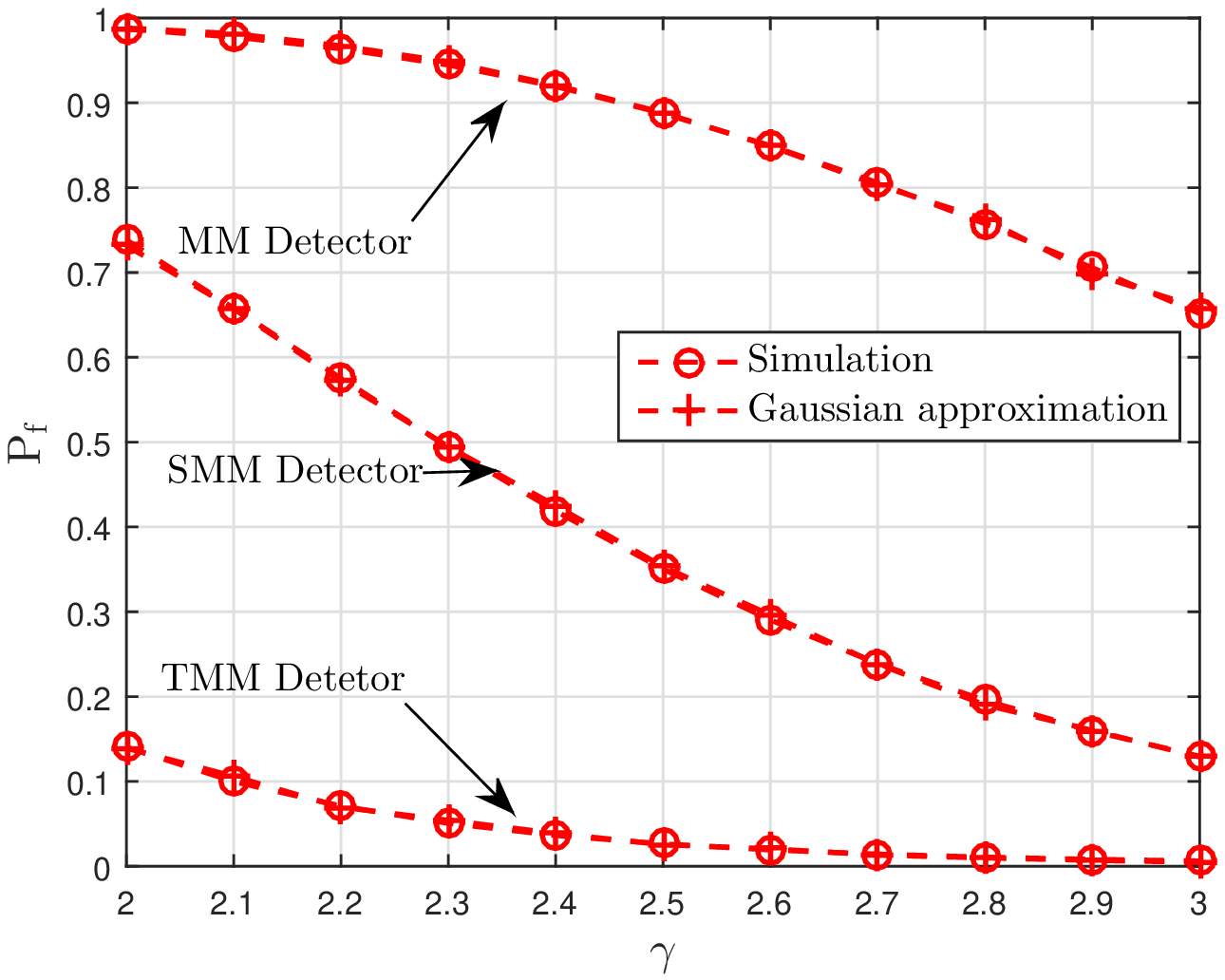}}
  \centerline{\hspace{10pt}(a)\hspace{160pt}(b)\hspace{160pt}(c)}\vspace{-5pt}
  \caption{Simulations of joint-PDF approximation and detection performance; (a) Analytical joint PDFs  under $NN_{\rm T} = 20$ (based on approximation approach with ${\rho _1} = 0.16$, ${\rho _2} = 0.26$, ${\rho _3} = 0.43$); (b) Simulated joint PDFs under $NN_{\rm T} = 20$  (based on empirical approach); (c) PF versus $\gamma$ under  ${\rho _1} = 0.16$, ${\rho _2} = 0.26$, ${\rho _3} = 0.43$, and $N_{\rm T}N=20$.}\vspace{10pt}
  \label{figure1}
\end{figure*}
  \subsection{Ordered Signal Detection}
 Three detectors are required to detect the possible number of signals on each block. In  the descending order of  eigenvalue values, we denote the predesigned detectors  respectively by (Maximum-Minimum ) MM detector, (Second-Maximum-Minimum) SMM detector,  (Third-Maximum-Minimum) TMM detector.
 We  formulate the normalized covariance matrix as $\widehat {\bf{R}} = \frac{1}{{{\sigma ^2}}}{\bf{Y}}{{\bf{Y}}^{\rm{H}}}\in {{\mathbb C}^{4 \times 4}}$ and suppose that the ordered eigenvalue of $\widehat {\bf{R}}$ are ${\lambda _1} > {\lambda _2} > {\lambda _3} > {\lambda _4}>0$. The test statistics are  therefore respectively denoted by
  \begin{equation}\label{E.12}
 T_{\rm MM} = \frac{{{\lambda _1}}}{{{\lambda _4}}},T_{\rm SMM} = \frac{{{\lambda _2}}}{{{\lambda _4}}}, T_{\rm TMM} = \frac{{{\lambda _3}}}{{{\lambda _4}}}
  \end{equation}
 A unified decision threshold, denoted by $\gamma$, is configured using
  \begin{equation}\label{E.13}
   T_{\rm MM}\mathop {\gtrless}\limits_{{{{\cal H}_0}}}^{{{{\cal H}_1}}}  \gamma, T_{\rm SMM}\mathop {\gtrless}\limits_{{{{{\overline{\cal H}}_2}}}}^{{{{\cal H}_2}}}  \gamma, T_{\rm TMM}\mathop {\gtrless}\limits_{{{{{\overline {\cal H}}_3}}}}^{{{{\cal H}_3}}}  \gamma
  \end{equation}
  where ${{\overline {\cal H}}_2}$  and ${{\overline {\cal H}}_3}$   represent  two alternative hypotheses that are respectively contrary to the hypothesis ${{{\cal H}}_2}$  and ${{{\cal H}}_3}$.  Therefore, identifying  the exact number of signals on one block can be achieved by an ordered detection and decision for composite  hypotheses. For example,  the existence of only one signal   is equivalent to successfully verifying ${{\overline {\cal H}}_3}$, then ${{\overline {\cal H}}_2}$ and finally  ${{{\cal H}}_1}$. Generally, the testing performance is measured by  the probability of detection (PD) and the probability of false alarm (PF) which are respectively denoted  for each detector  by
  \begin{equation}\label{E.14}
\begin{array}{l}
{\rm{P}}_{\rm{D}}^{{\rm{MM}}} = {\rm{Pr}}\left( {\left. {{H_1}} \right|{H_1}} \right),{\rm{P}}_{\rm{F}}^{{\rm{MM}}} = {\rm{Pr}}\left( {\left. {{H_1}} \right|{H_0}} \right)\\
{\rm{P}}_{\rm{D}}^{{\rm{SMM}}} = {\rm{Pr}}\left( {\left. {{H_{\rm{2}}}} \right|{H_{\rm{2}}}} \right),{\rm{P}}_{\rm{F}}^{{\rm{SMM}}} = {\rm{Pr}}\left( {\left. {{H_{\rm{2}}}} \right|{{{\overline {\cal H}}_2}}} \right)\\
{\rm{P}}_{\rm{D}}^{{\rm{TMM}}} = {\rm{Pr}}\left( {\left. {{H_{\rm{3}}}} \right|{H_{\rm{3}}}} \right),{\rm{P}}_{\rm{F}}^{{\rm{TMM}}} = {\rm{Pr}}\left( {\left. {{H_{\rm{3}}}} \right|{{{\overline {\cal H}}_3}}} \right)
\end{array}
 \end{equation}
 \begin{remark}
  The number of signals coexisting on one subcarrier block is not deterministic due to the random activation patterns and  cannot be predicted in advance. Therefore,  the common decision threshold $\gamma$  could guarantee that the number of  signals could be always precisely detected using a single threshold. Furthermore, this setup  ensures an analytical expression of $\gamma$ in the following.
  \end{remark}
 \subsection{ Determination of  $\gamma$}
 Examing Eq. (10), we stress that the first step is  to determine  the joint moments of two arbitrary eigenvalues.   Then we derive  the closed-form decision threshold $\gamma$  based on the  probability density function (PDF) approximated  from those moments.
  \subsubsection{Determination of  Moments}
Considering $\widehat {\bf{R}}$ under hypotheses ${{ {\cal H}}_0}$, we find that the joint distribution of  first eigenvalue and smallest one  is equivalent to that of  a Wishart matrix satisfying ${\cal C}{\cal W}\left( {{N_T}N,{{\bf{I}}_4}} \right)$~\cite{Goodman}.   Under ${{\overline {\cal H}}_i}$ for $2\le i\le3$,  the joint distribution of  $i$-th eigenvalue with smallest one  is equivalent to that of  the Wishart matrix.  Finally,  a closed  expression  of  joint moments of $\lambda_i$ and $\lambda_j$ is calculated by:
 \begin{equation}\label{E.15}
{\mathbb E}\left( {\lambda _i^m\lambda _j^n} \right) =\sum\limits_{\left\{ {a,b,c,d} \right\} \subseteq \left\{ {{A_1} \cup {A_2}} \right\}} {K{B_{a,b,c,d}}{f_{i,j,m,n}}\left( {a,b,c,d} \right)}
\end{equation}
where
  \begin{eqnarray}\label{E.16}
\left\{ {i,j,m,n} \right\} \subseteq \left\{ \begin{array}{l}
\left\{ {1,4,1,1} \right\},\left\{ {1,4,m,0} \right\},\left\{ {1,4,0,n} \right\},\\\left\{ {2,4,1,1} \right\},\left\{ {2,4,m,0} \right\},\left\{ {3,4,1,1} \right\},\\
\left\{ {3,4,m,0} \right\},1 \le m,n \le 2
\end{array} \right\}
  \end{eqnarray}
Here there exist ${{\cal A}_1}{\rm{ = }}\left\{ {\left\{ {a,b,c,d} \right\}\left| {\left\{ {a,b,c,d} \right\} \subseteq \left\{ {0 \cup {{\cal B}_1}} \right\}} \right.} \right\}$, and ${{\cal A}_2}{\rm{ = }}\left\{ {\left\{ {a,b,c,d} \right\}\left| {\left\{ {a,b,c,d} \right\} \subseteq } \right.{{\cal B}_2}} \right\}$. For ${{\cal B}_1}$ and ${{\cal B}_2}$,
\begin{equation}\label{E.17}
{{\cal B}_1}{\rm{ = }}\left\{ {{{\cal S}_{\left\{ {2,4,6} \right\}}},{{\cal S}_{\left\{ {2,5,5} \right\}}},{{\cal S}_{\left\{ {3,3,6} \right\}}},{{\cal S}_{\left\{ {3,4,5} \right\}}},{{\cal S}_{\left\{ {4,4,4} \right\}}}} \right\}
\end{equation}
  \begin{eqnarray}\label{E.18}
{{\cal B}_2}{\rm{ = }}\left\{ \begin{array}{l}
{{\cal S}_{\left\{ {1,1,4,6} \right\}}},{{\cal S}_{\left\{ {1,1,5,5} \right\}}},{{\cal S}_{\left\{ {1,2,3,6} \right\}}},{{\cal S}_{\left\{ {1,3,3,5} \right\}}},\\{{\cal S}_{\left\{ {1,2,4,5} \right\}}},{{\cal S}_{\left\{ {1,3,4,4} \right\}}},{{\cal S}_{\left\{ {2,2,2,6} \right\}}},{{\cal S}_{\left\{ {2,2,3,5} \right\}}},\\
{{\cal S}_{\left\{ {2,2,4,4} \right\}}},{{\cal S}_{\left\{ {2,3,3,4} \right\}}},{{\cal S}_{\left\{ {3,3,3,3} \right\}}}
\end{array} \right\}
  \end{eqnarray}
  where ${{\cal S}_{\left\{ {\cdot} \right\}}}$ represents the permutation of the elements of set.
  \begin{figure*}[ht]
  \begin{eqnarray}\setcounter{equation}{17}
  \label{E.21}
 \hspace{-10pt}
{f_{2,4,m,0}}\left( {a,b,c,d} \right)\hspace{-5pt}&=&\hspace{-5pt} \sum\limits_{{k_1} = 0}^{N + a - 4} {\frac{{{K_2}\left( {{{\overline k }_1} + {k_1} + m} \right)!{{\overline k }_1}!}}{{{k_1}!{2^{{{\overline k }_2} + {k_1} + m + 1}}}}}  - \sum\limits_{{k_1} = 0}^{N + a - 4} {\sum\limits_{{k_2} = 0}^{{{\overline k }_2}} {\frac{{{K_2}\left( {{{\overline k }_2} + {k_1} + {k_2} + m} \right)!{{\overline k }_1}!}}{{{k_1}!{k_2}!{3^{{{\overline k }_2} + {k_1} + {k_2} + m + 1}}}}} } \nonumber \\
  \hspace{-5pt}&-&\hspace{-5pt}\sum\limits_{{k_1} = 0}^{N + a - 4} {\sum\limits_{{k_2} = 0}^{N + d - 4} {\frac{{{K_2}\left( {{{\overline k }_2} + {k_1} + m} \right)!\left( {{{\overline k }_1} + {k_2}} \right)!}}{{{k_1}!{k_2}!{2^{{{\overline k }_1} + {k_2} + 1}}{2^{{{\overline k }_2} + {k_1} + m + 1}}}}} }   \nonumber +\sum\limits_{{k_1} = 0}^{N + a - 4} {\sum\limits_{{k_2} = 0}^{N + d - 4} {\sum\limits_{{k_3} = 0}^{{{\overline k }_2} + {k_2}} {\frac{{{K_2}\left( {{{\overline k }_1} + {k_2}} \right)!\left( {{{\overline k }_2} + {k_1} + {k_3} + m} \right)!}}{{{k_1}!{k_2}!{k_3}!{2^{{{\overline k }_2} + {k_{\rm{2}}} - {k_{\rm{3}}} + 1}}{4^{{{\overline k }_2} + {k_1} + {k_3} + m + 1}}}}} } }\\
  \vspace{-30pt}
   \end{eqnarray}
   \end{figure*}
   \begin{figure*}[ht]
\begin{equation}\setcounter{equation}{18}
\label{E.22}
{f_{3,4,m,0}}\left( {a,b,c,d} \right) = \sum\limits_{{k_1} = 0}^{N + a - 4} {\sum\limits_{{k_2} = 0}^{{{\overline k }_2} + {k_1}} {\frac{{{K_2}\left( {{{\overline k }_1} + {k_2} + m} \right)!\left( {{{\overline k }_2} + {k_1}} \right)!}}{{{k_1}!{k_2}!{2^{{{\overline k }_2} + {k_1} - {k_2} + 1}}{3^{{{\overline k }_1} + {k_2} + m + 1}}}}} }-\sum\limits_{{k_1} = 0}^{N + a - 4} {\sum\limits_{{k_2} = 0}^{{{\overline k }_2} + {k_1}} {\sum\limits_{{k_3} = 0}^{N + d - 4} {\frac{{{K_2}\left( {{{\overline k }_1} + {k_2} + {k_3} + m} \right)!\left( {{{\overline k }_2} + {k_1}} \right)!}}{{{k_1}!{k_2}!{k_3}!{2^{{{\overline k }_2} + {k_1} - {k_2} + 1}}{4^{{{\overline k }_1} + {k_2} + {k_3} + m + 1}}}}} } }
\end{equation}
\vspace{-25pt}
\end{figure*}
 \begin{figure*}[ht]
\begin{eqnarray}\setcounter{equation}{19}
\label{E.19}
\hspace{-10pt}
{f_{2,4,1,1}}\left( {a,b,c,d} \right)  \hspace{-5pt}&=&\hspace{-5pt} \sum\limits_{{k_1} = 0}^{{{\overline k }_3} - m} {\sum\limits_{{k_2} = 0}^{{{\overline k }_1}} {\sum\limits_{{k_3} = 0}^{{k_1} + {{\overline k }_2} + m} {\frac{{\left( {{{\overline k }_3} - m} \right)!\left( {{k_1} + {{\overline k }_2} + m} \right)!{{\overline k }_3}!\left( {{k_4} - {k_1} - 1} \right)!}}{{{k_1}!{k_2}!{k_3}!{2^{N + b + {k_{\rm{1}}} - {k_{\rm{3}}} + m - 3}}{4^{N + d + {k_2} + {k_3} + m - 3}}}}} } } \nonumber\\
   \hspace{-5pt}&-&\hspace{-5pt}\sum\limits_{{k_1} = 0}^{{{\overline k }_3} - m} {\sum\limits_{{k_2} = 0}^{{{\overline k }_1}} {\sum\limits_{{k_3} = 0}^{{{\overline k }_2} + {k_1} + {k_2} + m} {\frac{{\left( {{{\overline k }_3} - m} \right)!\left( {{k_1} + {{\overline k }_2} + {k_2} + m} \right)!{{\overline k }_1}!\left( {{k_4} - {k_1} - {k_2} - 1} \right)!}}{{{k_1}!{k_2}!{k_3}!{3^{N + b + {k_{\rm{1}}} + {k_2} - {k_3} + m - 3}}{4^{N + d + {k_3} + m - 3}}}}} } }
\end{eqnarray}
\vspace{-15pt}
 \end{figure*}
  \begin{figure*}[ht]
\begin{equation}\setcounter{equation}{20}
\label{E.20}
{f_{3,4,1,1}}\left( {a,b,c,d} \right) = \sum\limits_{{k_1} = 0}^{{{\overline k }_3} - m} {\sum\limits_{{k_2} = 0}^{{{\overline k }_2} + {k_1}} {\sum\limits_{{k_3} = 0}^{{{\overline k }_1} + {k_2} + m} {\frac{{\left( {{{\overline k }_3} - m} \right)!\left( {{{\overline k }_2} + {k_1}} \right)!\left( {{{\overline k }_1} + m + {k_2}} \right)!\left( {{k_4} - {k_1} - {k_2} - 1} \right)!}}{{{k_1}!{k_2}!{k_3}!{2^{{{\overline k }_2} + {k_{\rm{1}}} - {k_2} + 1}}{3^{{{\overline k }_1} + {k_2} - {k_3} + m + 1}}{4^{{k_4} - {k_1} - {k_2}}}}}} } }
\end{equation}
\hrulefill
\end{figure*}
 \begin{figure*}[hb]
 \hrulefill
\begin{equation}\setcounter{equation}{21}
\label{E.18.1}
{{B}_{a,b,c,d}} = \left\{ {\begin{array}{*{20}{c}}
{ - 6}&{\left\{ {a,b,c,d} \right\} \subseteq \left\{ {{{\cal S}_{\left\{ {0,4,4,4} \right\}}} \cup {{\cal S}_{\left\{ {2,3,3,4} \right\}}}} \right\}}\\
{ - 4}&{\left\{ {a,b,c,d} \right\} \subseteq {{\cal S}_{\left\{ {1,3,3,5} \right\}}}}\\
1&{\left\{ {a,b,c,d} \right\} \subseteq {{\cal S}_{\left\{ {0,2,4,6} \right\}}}}\\
2&{\left\{ {a,b,c,d} \right\} \subseteq \left\{ {{{\cal S}_{\left\{ {0,3,4,5} \right\}}} \cup {{\cal S}_{\left\{ {1,2,3,6} \right\}}}} \right\}}\\
4&{\left\{ {a,b,c,d} \right\} \subseteq \left\{ {{{\cal S}_{\left\{ {1,1,5,5} \right\}}} \cup {{\cal S}_{\left\{ {1,3,4,4} \right\}}} \cup {{\cal S}_{\left\{ {2,2,3,5} \right\}}} \cup {{\cal S}_{\left\{ {2,2,4,4} \right\}}}} \right\}}\\
{24}&{\left\{ {a,b,c,d} \right\} \subseteq {{\cal S}_{\left\{ {3,3,3,3} \right\}}}}\\
{ - 2}&{\rm otherwise}
\end{array}} \right.
\end{equation}
 \end{figure*}
 The specific function $f$ can be  shown from Eq.~(\ref{E.21}) to  Eq.~(\ref{E.20})
where  ${\overline k _1} = N{N_{\rm{T}}} + c - 4$, ${\overline k _2} = N{N_{\rm{T}}} + b - 4$, ${\overline k _3} = N{N_{\rm{T}}} + a + m - 4$, ${k_4} = N{N_{\rm{T}}} + d + {k_1} + {k_2} + {k_3} + m - 3$ and   ${K_2} = \Gamma \left( {N{N_{\rm{T}}} + a - 3} \right)\Gamma \left( {N{N_{\rm{T}}}+ d - 3} \right)$.  $\Gamma \left( \cdot \right)$ is the Gamma function. Note that the other function  ${f_{1,4,1,1}}\left( {a,b,c,d} \right)$, ${f_{1,4,0,n}}\left( {a,b,c,d} \right)$ and  ${f_{1,4,m,0}}\left( {a,b,c,d} \right)$ can be found in~\cite{Shakir}.  $K$ satisfies $K = \frac{1}{{12\left( {{N{N_{\rm{T}}}} - 1} \right)!\left( {{N{N_{\rm{T}}}} - 2} \right)!\left( {{N{N_{\rm{T}}}} - 3} \right)!\left( {{N{N_{\rm{T}}}} - 4} \right)!}}$ and ${{B}_{a,b,c,d}}$ is shown in Eq.~(\ref{E.18.1}).
\subsubsection{ Analytical Solution of $\gamma$}
 Inspired from the method in~\cite{Shakir}, we consider the joint PDF of any eigenvalue (except the smallest one) and  the smallest one can be approximated similarly, that is, for $1\le i \le 3$,
  \begin{equation}\label{E.25}
   {f_{{\lambda _i},{\lambda _4}}}\left( {x,y} \right) = \frac{1}{{2\pi {\xi _{{\lambda _i}}}{\xi _{{\lambda _4}}}\sqrt {1 - {\rho_{i}^2}} }}\exp \left\{ { - \frac{\kappa_i }{{2{{\left( {1 - \rho_i } \right)}^2}}}} \right\}
    \end{equation}
  where  ${\xi _{{\lambda _i}}}$ denote  the standard deviation of the eigenvalue ${{\lambda _i}}$ and can be derived in Eq.~(\ref{E.15}). $ \rho_i$ is the correlation coefficient between $\lambda _i$ and ${\lambda _4}$.
  The parameter $ \rho_i$  is  given by $ \rho_i  = \frac{{{\zeta _{{\lambda _i},{\lambda _4}}} - {\zeta _{{\lambda _i}}}{\zeta _{{\lambda _4}}}}}{{{\xi _{{\lambda _i}}}{\xi _{{\lambda _4}}}}}, 1\le i \le 3$ and  $\kappa_i $ is extended as:
  \begin{equation}\label{E.26}
     \kappa_i  = \frac{{{{\left( {x - {\zeta _{{\lambda _i}}}} \right)}^2}{\rm{ + }}2\rho_i {\xi _{{\lambda _i}}}{\xi _{{\lambda _4}}}\left( {x - {\zeta _{{\lambda _i}}}} \right)\left( {y - {\zeta _{{\lambda _4}}}} \right){\rm{ + }}{{\left( {y - {\zeta _{{\lambda _4}}}} \right)}^2}}}{{\xi _{{\lambda _i}}^2\xi _{{\lambda _4}}^2}}
   \end{equation}
   where ${\zeta _{{\lambda _i}}}$ denote the expectation of ${\lambda _i}$ and ${\zeta _{{\lambda _i},{\lambda _4}}}$ represent  the expectation of two-variate variable  ${\lambda _i}$ and  ${\lambda _4}$.

   More accurately, we compare  the joint PDF generation using  the approximation approach with that using the empirical approach by simulations  under $NN_{\rm T} = 20$.   Specifically, Fig.~\ref{figure1}(a) illustrates the one using approximation  approach whereas Fig.~\ref{figure1}(b) shows the one based on the empirical approach. We can see that the PDFs under two methods are almost in  agreement  provided that  the mean and the variance of the eigenvalues and correlation between them can be obtained.

Based on the approximated PDFs and given threshold  $\gamma_{i}$,  the cumulative distribution function (CDF) of the ratio between  ${\lambda _i}$ and  ${\lambda _4}$, denoted by $F_{i}\left( \gamma_{i}  \right)$,  can be expressed by $F_{i}\left( \gamma_{i}  \right) = \Phi \left\{ {\frac{{{\zeta _{{\lambda _4}}}\gamma_{i}  - {\zeta _{{\lambda _i}}}}}{{{\xi _{{\lambda _i}}}{\xi _{{\lambda _4}}}\chi \left( \gamma_{i}  \right)}}} \right\},     \chi \left( \gamma _{i} \right) = \sqrt {\frac{{{\gamma_{i}^2}}}{{\xi _{{\lambda _i}}^2}} - \frac{{2\rho_{i} \gamma_{i} }}{{{\xi _{{\lambda _i}}}{\xi _{{\lambda _4}}}}} + \frac{1}{{\xi _{{\lambda _4}}^2}}}$. Here $\Phi\left\{ \cdot \right\}$ denotes  CDF of a standard Gaussian random variable.  We then can determine  ${\gamma _i}$  by $\gamma_{i}  \buildrel \Delta \over = f_{i}\left( {N{N_{\rm{T}}}} \right) =  \frac{{{\zeta _{{\lambda _i}}}{\zeta _{{\lambda _4}}} - \tau _i^2{\rho _i}{\xi _{{\lambda _i}}}{\xi _{{\lambda _4}}} + {\tau _i}\sqrt {{\delta _i} - 2{\rho _i}{\xi _{{\lambda _i}}}{\xi _{{\lambda _4}}}{\zeta _{{\lambda _i}}}{\zeta _{{\lambda _4}}}} }}{{\zeta _{{\lambda _4}}^2 - \tau _i^2\xi _{{\lambda _4}}^2}}$
 where ${\delta _i} = \zeta _{{\lambda _i}}^2\xi _{{\lambda _4}}^2 + \zeta _{{\lambda _4}}^2\xi _{{\lambda _i}}^2 + \left( {\rho _i^2 - 1} \right)\tau _i^2\xi _{{\lambda _i}}^2\xi _{{\lambda _4}}^2$, $\tau_{i}  = {\Phi ^{ - 1}} \left\{ F_{i}\left( \gamma_{i} \right)\right\}$.

The optimal $\gamma$ should  make all the expression forms of PF in Eq.~(\ref{E.14}) approach the values that are  as small as possible.
\begin{theorem}
Given an  upper bound of PF, denoted by  $P$ and with arbitrary value, the  decision threshold $\gamma$ able to guarantee $1- F_{i}\left( \gamma  \right)\le P,  \forall  1\le i \le 3 $  is given by:
  \begin{equation}\label{E.29}
\gamma  = \max \left\{ {{\gamma^{*} _1},{\gamma^{*} _2},{\gamma^{*} _3}} \right\}
 \end{equation}
 where ${\gamma^{*} _i}$ satisfies ${F_i}\left( {{\gamma^{*} _i}} \right) =1-P$.  And  ${N{N_{\rm{T}}}}$ achieving $\gamma$ can be calculated according to:
  \begin{equation}\label{E.30}
  N{N_{\rm{T}}} = {f^{ - 1}_{{i_{{\rm{opt}}}}}}\left( \gamma  \right)
   \end{equation}
   where ${i_{{\rm{opt}}}} = \mathop {\arg \max }\limits_i {\gamma^{*} _i}$
 \end{theorem}
 The verification of this theorem is easy since  ${F_i}$ is a monotonically-increasing function of $\gamma$.
 We compare the PF performance of three detectors in Fig.~\ref{figure1}(c) where   two different approaches are  respectively simulated, that is,  the Monte Carlo simulation and Gaussian approximation.  $NN_{\rm T}$ is configured  to be 20. As shown in the figure,   PF curves using theoretical approximation match well with those under practical simulation. Three types of PF gradually decreases to be zero as well,  with the increase of  $\gamma$.
 \subsection{ Formulation of  Detection Based Binary Number System}
 Basically,  theorem 1 provides a  quantitative method for measuring how many  subcarriers are required in one block for precise coding with zero PF and perfect  PD. Therefore, we have the following proposition:
 \begin{proposition}
The number $N^{*}$ of  subcarriers in one subcarrier block that are  enabled to precisely carry binary number information  can be  calculated from   the Eq.~(\ref{E.30}) by configuring  $P$ to be an arbitrarily small value.
\end{proposition}
 Generally, when we define the total number of subcarriers allocated for channel estimation  as $N_{\rm Total}$, typically equal to several hundreds,    $N^{*}$ satisfies $N^{*}={{{N_{{\rm{Total}}}}} \mathord{\left/
 {\vphantom {{{N_{{\rm{Total}}}}} B}} \right.
 \kern-\nulldelimiterspace} B}$.

 To verify the proposition, let us turn to a 3D plot of PD and PF versus the threshold $\gamma$ and $NN_{\rm T}$  in Fig.~\ref{figure2}(a). As we can see, with  $NN_{\rm T}=60, \gamma=3$, PF is equal to zero  while the PD is always maintained  to be 1 for all the three detectors. In this sense, $N^{*}N_{\rm T}=60$ is enough for precise coding when $\gamma=3$. A control variable $N_{\rm B}$ is defined by
 \begin{equation}\label{E.31}
{N_{\rm{B}}}{ \buildrel \Delta \over =}N^{*}N_{\rm T}={{{N_{{\rm{Total}}}}{N_{\rm{T}}}} \mathord{\left/
 {\vphantom {{{N_{{\rm{Total}}}}{N_{\rm{T}}}} B}} \right.
 \kern-\nulldelimiterspace} B}
\end{equation}
 Therefore, we know that any block configured with $N_{\rm B}\ge  60$ can carry binary number information precisely. When each block satisfies those requirements,  each code  digit ${c_i}$ that corresponds to the $i$-th subcarrier block is endowed with the following binary number:
 \begin{equation}\label{E.32}
{c_i} = \left\{ {\begin{array}{*{20}{c}}
1&{if \,\,\,there \,\,\, exist  \,\,\, signals}\\
0&{otherwise}
\end{array}} \right.
  \end{equation}
\section{Key Technique II: Code Construction}
\label{CC}
On the basis of  coded subcarrier blocks, we, in this section,  develop a  CFBG coding theory to construct a  binary group code with  fixed length and constant weight.

\subsection{ Binary Arithmetic Rule Between Codewords}
 The  binary arithmetic  rule between any two codewords to be designed is necessary and  should be able to represent the overlapping operation  precisely.  Intuitively, two rules on the code-frequency domain  can be identified and  mathematically interpreted as follows:
\newtheorem{definition}{Definition}
\newtheorem{principle}{Principle}
\begin{definition}
\label{definition1}
The superposition (SP) sum  ${\bf{z}} = {\bf{x}}{{V}}{\bf{y}}$ (designated as the digit-by-digit Boolean sum) of two $B$-dimensional
binary vectors ${\bf{x}} = \left( {{x_1},{x_2}, \ldots ,{x_{B}}} \right)$  and ${\bf{y}} = \left( {{y_1},{y_2}, \ldots ,{y_{B}}} \right)$ is defined by:
  \begin{equation}\label{E.33}
{z_i} = \left\{ {\begin{array}{*{20}{c}}
0&{if\,\,{x_i} = {y_i} = 0}\\
1&{otherwise}
\end{array}} \right., \forall  1\le i\le {B}
  \end{equation}
  and we say that a binary vector $\bf x$ includes a binary vector $\bf y$ if the Boolean sum satisfies  ${\bf{y}}{{V}}{\bf{x}}={\bf{x}} $
  \end{definition}
  \begin{definition}
  The algebraic superposition (ASP) sum (designated as the digit-by-digit  sum)  is defined by  ${\bf{d}} = {\bf{x}}{{V}}{\bf{y}}$ in which two $B$-dimensional binary vectors ${\bf{x}} = \left( {{x_1},{x_2}, \ldots ,{x_{B}}} \right)$  and ${\bf{y}} = \left( {{y_1},{y_2}, \ldots ,{y_{B}}} \right)$ satisfy:
  \begin{equation}\label{E.34}
{{ d}_i} ={{ x}_i} + {{ y}_i}, \forall  1\le i\le {B}
  \end{equation}
 \end{definition}
\subsection{Coding Principle}
Establishment of CFBG coding on the basis of  the formulated binary code  requires us to thoroughly  analyze the issues induced by  specific superposition rules. We attempt  to achieve pilot conveying  using binary code satisfying  the principle of SP sum. In this case, several key coding principles  constrained by  Problem 2 can be identified.

First,  we have to admit that Eva can launch  randomly-imitating  attack, namely, selecting randomly one codeword in  the same publicly-known code as Bob (or Charlie)  and  activating  subcarrier blocks  as the codeword indicates. Therefore, we hope to  guarantee that each  superposition of up to  three different codewords  is  unique and each  superimposed codeword can be uniquely and correctly  decomposed   into original codewords.
To achieve this, we propose  two principles:
\begin{principle}
\label{Proposition3}
Every sum of up to three different codewords  can be  decomposed by   no codeword other than those used to form the
sum.
\end{principle}
\begin{principle}
\label{Proposition4}
Every sum of up to three  different codewords is distinct from every other sum of three  or fewer codewords.
\end{principle}
 \begin{remark}
 Thanks to  the  sum operation for up to three  different codewords,   proposed CFBG method applies to single-user scenario threaten by  an additional attacker Eva.
 \end{remark}
The  code satisfying the two principles can be divided by two independent codes, respectively for Bob and Charlie, thus  distinguishing their own codewords from each other.    In this way, the silence of Eva, if exists, can  be detected since any superimposition from extra codewords will induce a complete new observation codeword at Alice,  therefore indicating the existence of an attack.

Randomly-imitating  attack   is also enabled to be detected perfectly since any extra  duplicate  codeword  can transform original superposition codewords ( superimposed by Bob and Charlie) into a novel  distinct and  identifiable codeword in the code. This is determined by the two principles, which is, however, unreliable when a wideband jamming attack happens.
\begin{problem}
When a wideband jamming attack happens, the interpreted codeword at Alice is  a vector with all elements``1''  which  carry no  information useful for  Alice. The  codewords decomposed from the superimposed codeword under the attack  may also belong to the code. In this case, Alice will ignore  jamming attack and make a wrong decision that there exists no attack.
\end{problem}
To solve this issue, we reconsider the two principles and  discover  an important property, that is,
\begin{property}
\label{Proposition5}
For every sum of up to three different codewords within the code, if  we reduce   by one any codeword digit indicating single  signal on the subcarrier block, the resulted codeword  can be decomposed by no codeword in the code.
\end{property}
This requires the previous  block detection technique   combined with the process of  code design. We stress that   this property can resolve Problem 4 and its effectiveness  can be interpreted by the following fact:
\begin{fact}
Under jamming attack, any superimposed digit  indicating  single  signal on a subcarrier block logically suggests  that the  digits previously exploited by Bob and Charlie at the same position are both  of zero value. Therefore, if we reduce those digits to be zero, the weight of interpreted codewords is unchanged.  The interpreted codewords  will belong to the original codebook. Otherwise  when there is no jamming attack, we can know that  there exist non-zero digits  exploited by  Bob and/or Charlie  at the same digit positions and  any reduction of those digits will induce the interpreted codeword of less weight as well.     In this case, the interpreted codeword will never belong to  the predesigned  code and  finally we can distinguish  whether jamming attack happens.
\end{fact}
In summary, the two principles not only guarantee the identification and classification for  hybrid attack, but also  provide the  basic functionalities of codeword conveying, separation and  identification.  Obviously, those  principles combined with  block detection technique  constitute  the core of  CFBG code.

   \begin{figure*}
 \centerline{
  ~\includegraphics[width=2.50in]{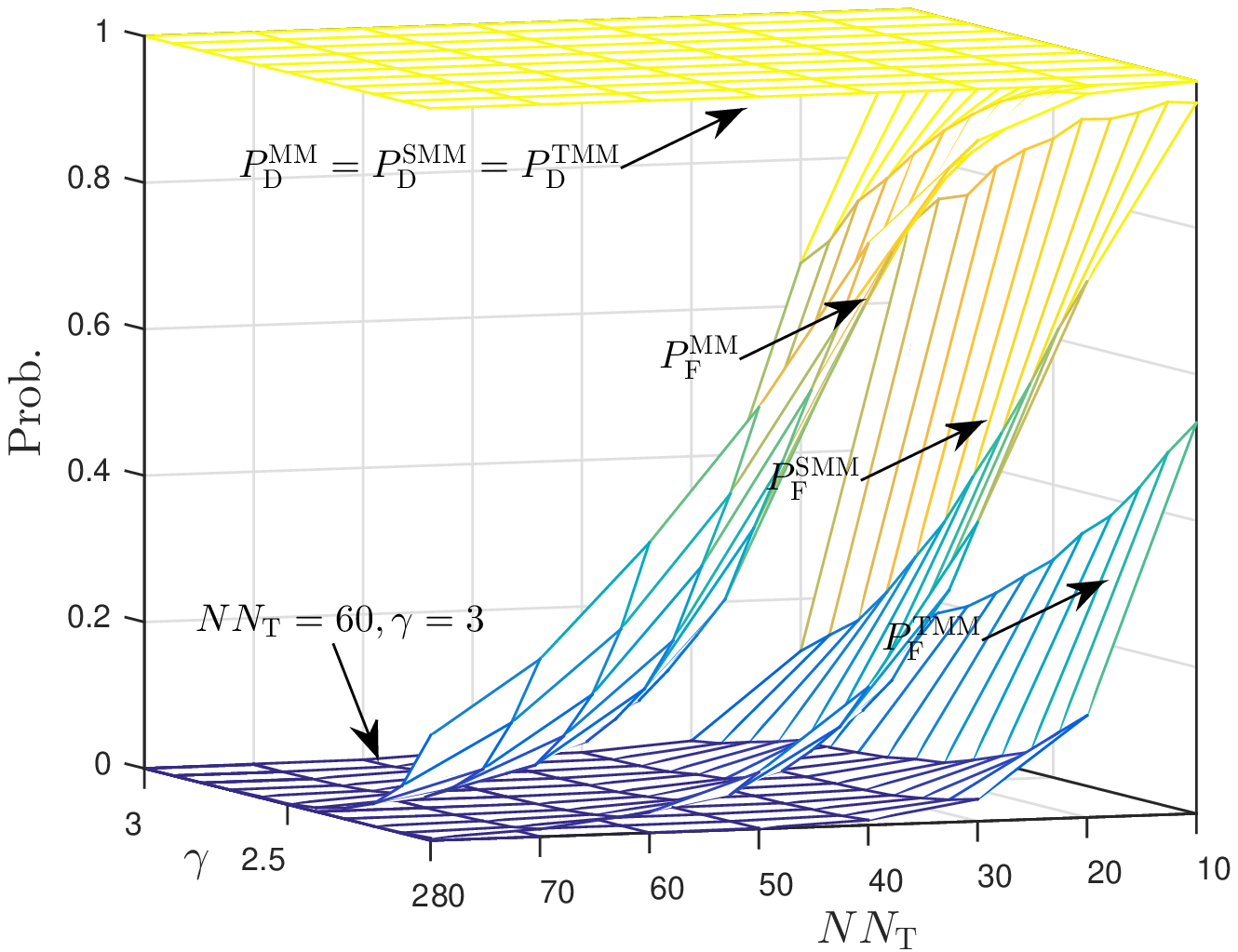} \hspace{-10pt} \includegraphics[width=2.50in]{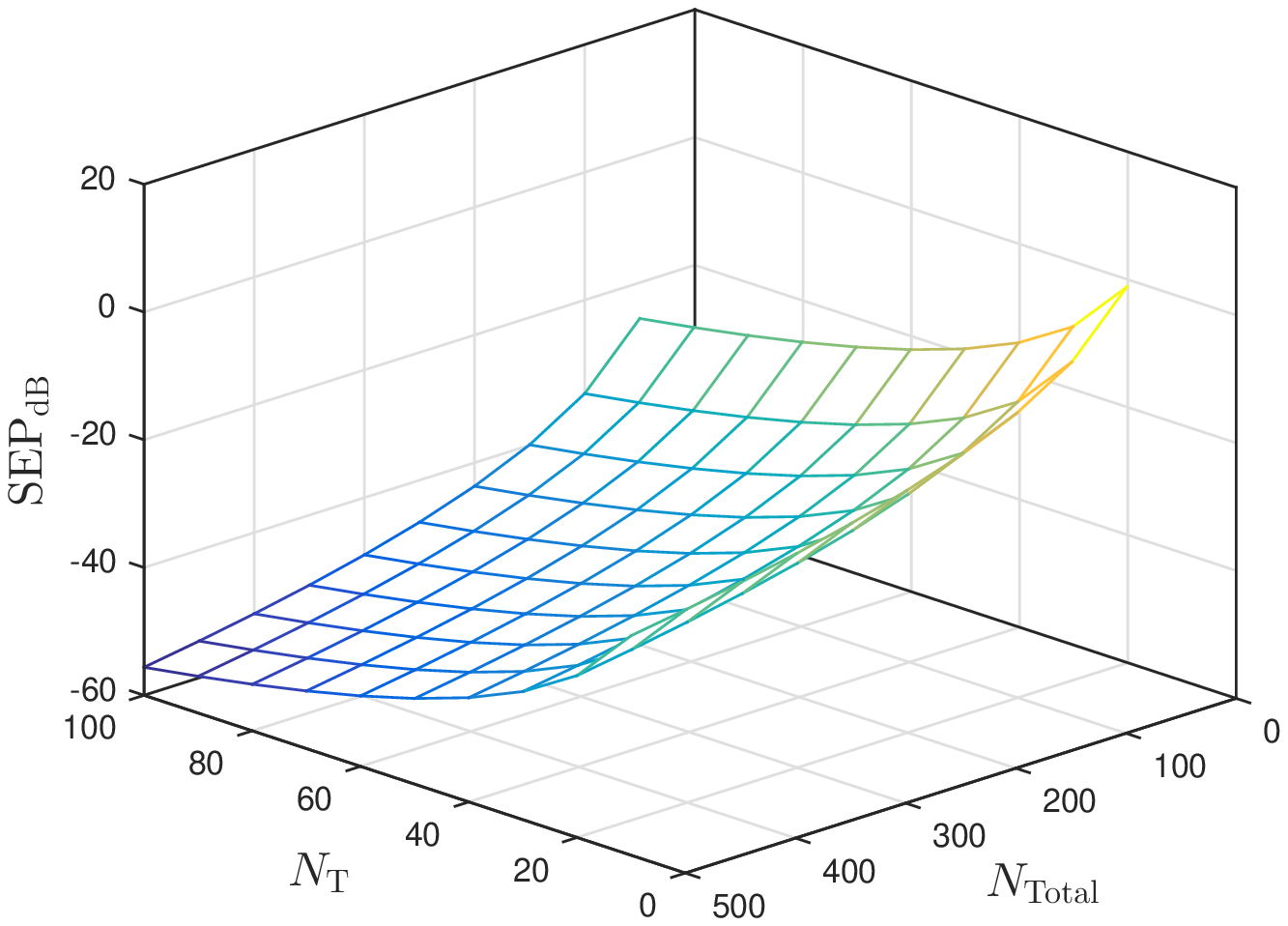}\hspace{-10pt}  \includegraphics[width=2.50in]{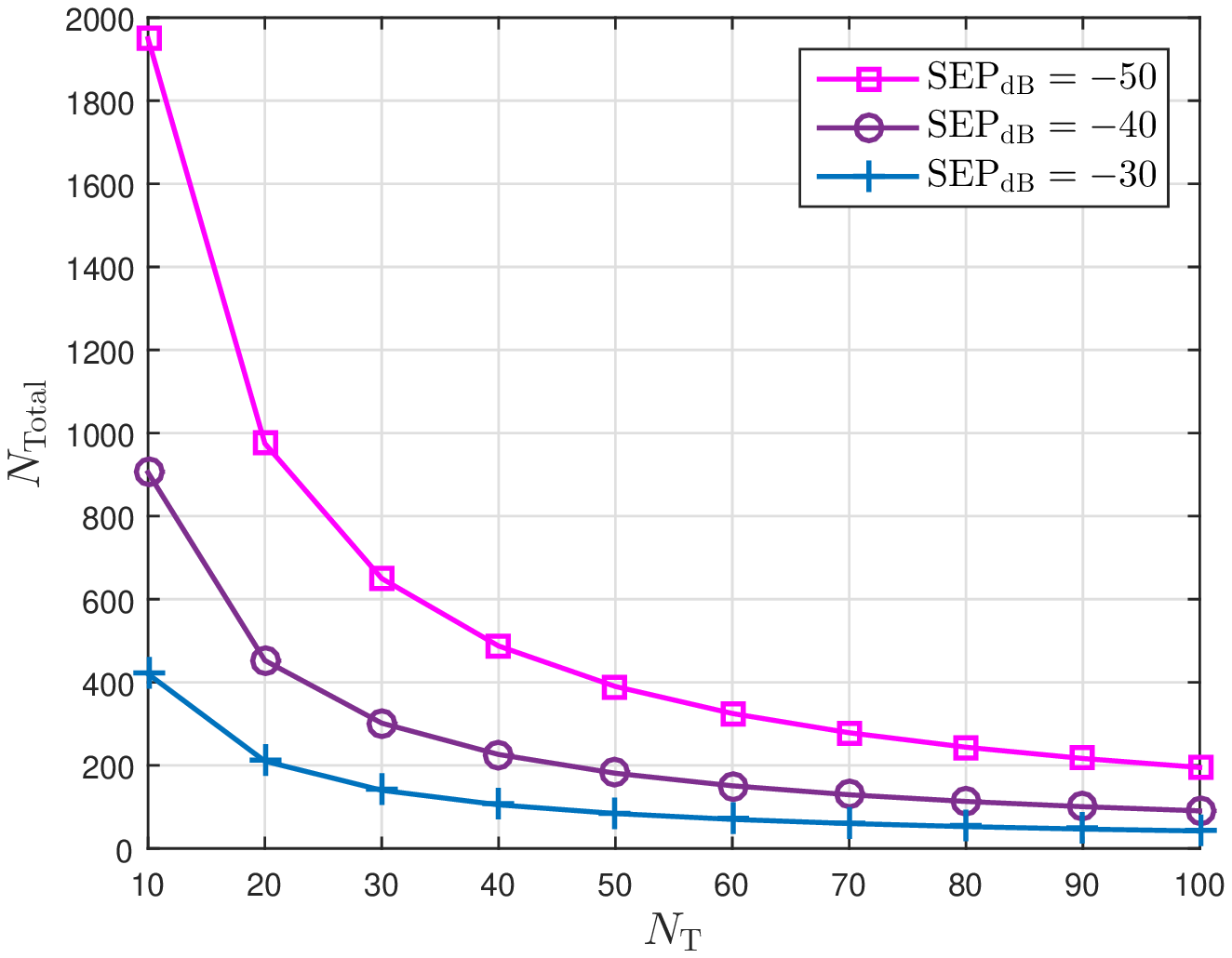}}
  \centerline{\hspace{10pt}(a)\hspace{160pt}(b)\hspace{160pt}(c)}\vspace{-5pt}
  \caption{ Detection performance and CFBG codebook  performance; (a) PD  and PF  versus different $N_{\rm T}N$ and $\gamma$. Note that we can configure $N=1$, that is, one subcarrier in one block, and in this sense, the curves indicate the performance of PD and PF versus $\gamma $ and $N_{\rm T}$; (b) SEP performance versus $N_{\rm T}$ and $N_{\rm Total}$; (c) Tradeoff curves:  $N_{\rm Total}$ versus  $N_{\rm T}$  under fixed SEP.}\vspace{10pt}
  \label{figure2}
\end{figure*}

\subsection{Construction of CFBG Codebook}
 First, we construct the codebook satisfying the above two principles through a well-known ZFD code  proposed in~\cite{Kautz}. Its definition is given as follows:
 \begin{definition}
 A  ZFD code  $\cal C$ with order $m$ is defined by a collection of $C$ $B$-dimensional binary vectors ${\bf c}_{i}, 1\le i\le C$ for which no SP sum ${\bf c}_1{{V}}{\bf c}_2 {{V}} ... {{V}} {\bf c}_k$ of $k \le m$  codewords includes any other codeword not used in this sum.
 \end{definition}
 Intuitively, the definition of  ZFD code  satisfies  the  Principle 1.  Based on the definition,  we know that any sum ${\bf c}_1{{V}}{\bf c}_2 {{V}} ... {{V}} {\bf c}_k$ of $k \le m$  codewords cannot  include the superposition sum of any other codewords, for instance, ${\bf c}_{j_1}{{V}}{\bf c}_{j_2} {{V}} ... {{V}} {\bf c}_{j_k}, k \le m$ with $\left\{ {{{1,}} \ldots {{,k}}} \right\} \ne \left\{ {{{{j}}_1}{{,}} \ldots {{,}}{{{j}}_k}} \right\}$,  since each of the other codewords cannot be included in the sum.   Therefore, Principle 2 can be guaranteed by:
  \begin{proposition}
 For a  ZFD code  $\cal C$ with order $m$,  two arbitrary  SP sums  each of which is superimposed by $k \le m$ code words are identical if and only if the two  codeword sets respectively constituting the two sums are completely identical as well.
 \end{proposition}
Second,  we  define  the concept of BD code  satisfying  Property 1 by the following:
  \begin{definition}
 A  BD code $\cal B$  is the one that has the same codewords as ZFD code but follows ASP sum principle.
 \end{definition}
Finally, we focus on the construction of  two codes and show how to construct the  CFBG codebook   by integrating the two codes.  Let us begin by introducing the ZFD code construction.
\subsubsection{Relationship between maximum-distance separable (MDS) and ZFD code and theoretical results} An arbitrary $m$-order ZFD code is constructed  on the basis of $m$-order MDS code. We consider a $m$ order $B$-digit ZFD code  with the constant weight.
 Generally, MDS code is an efficient way  to construct the ZFD code since  each of codewords  belonging to MDS code exactly occurs once in the overall code set~\cite{Singleton}.  Basically,  MDS code is a $q$-nary error-correcting code whose codeword digits are members of a set of $q$ basic symbols. MDS code has the maximum possible distance $d=r+1$ for given code size $C=q^{k}$ and codeword length $n=k+r$. A $m$ order $B$-digit  ZFD code can be constituted  from a $q$-nary ${B \mathord{\left/
 {\vphantom {B q}} \right.
 \kern-\nulldelimiterspace} q}$-digit MDS code by representing each digit of codeword with  a unique weight-one binary $q$-tuple. For example,  the $q$-nary symbols 0, 1, . . . $q-1$ are to be replaced by the $q$-digit binary vectors $1,0, \ldots, 0$, $0,1, \ldots ,0$, $0,0, \ldots ,1$ respectively.  In this context, the code size satisfies the following relationships:
\begin{equation}\label{E.35}
C=q^{k} ,m =  {\frac{{{{n} - 1}}}{{{k} - {1}}}}, q \ge m\left( {k - 1} \right) \ge 3, n=B \mathord{\left/
 {\vphantom {B q}} \right.
 \kern-\nulldelimiterspace} q
\end{equation}
where $m \in {\mathbb Z^{+}}, n \in{ \mathbb Z^{+}}$. Furthermore,   $m = 3,B = \frac{{{N_{{\rm{Total}}}}{N_{\rm{T}}}}}{{{N_B}}}$ has to be satisfied given three wireless nodes at most. The constraint of  $ 2 \le k \le q - 1,2 \le r \le q - 1$ is also imposed due to MDS property~\cite{Singleton}.
\begin{theorem}
The size of  MDS based ZFD code  satisfies:
\begin{equation}\label{E.36}
C = {q^k},\frac{{{N_{\rm{T}}}}}{{{N_B}}} = \frac{{q\left( {3k - 2} \right)}}{{{N_{{\rm{Total}}}}}}, 2 \le k \le \frac{{q + 3}}{3},q \ge 3
\end{equation}
\end{theorem}
\begin{proof}
First, since $m = \frac{{n - 1}}{{k - 1}} = 3$   and $n = {B \mathord{\left/
 {\vphantom {B q}} \right.
 \kern-\nulldelimiterspace} q}$ , we can easily derive $\frac{{{N_{\rm{T}}}}}{{{N_B}}} = \frac{{q\left( {3k - 2} \right)}}{{{N_{{\rm{Total}}}}}}$.  Then we focus on the range of parameter $k$. Combing $q \ge m\left( {k - 1} \right) \ge 3$  with  $2 \le k \le q - 1$, we can derive $2 \le k \le \frac{{q + 3}}{3}$  when $q \ge 3$ . Furthermore, we consider the constraint  $2 \le r \le q - 1$. Since $r = n - k = 2k - 2$, we have  $k \le \frac{{q + 1}}{2}$. Comparing the upper bound $\frac{{q + 3}}{3}$   and  $\frac{{q + 1}}{2}$, we can finally determine  the range of  $k$ satisfying $2 \le k \le \frac{{q + 3}}{3}$ given  $q \ge 3$.  The theorem is proved.
\end{proof}
 \begin{remark}
For single-user scenario,  $m$ is set to be 2 and the above method still holds true but with  the consideration of  different parameter configurations. In this case, we can easily  have $C = {q^k},\frac{{{N_{\rm{T}}}}}{{{N_B}}} = \frac{{q\left( {2k - 1} \right)}}{{{N_{{\rm{Total}}}}}}, 3 \le k \le \frac{{q +2}}{2},q \ge 4$.
 \end{remark}
 Based on the above theoretical support, we aim to  construct the ZFD code in details.
 \subsubsection{Construction of ZFD code}
 Firstly, we exploit the concept of  Latin hypercubes defined in the following.
\begin{definition}
A Latin $k$-dimensional cube of order $q$ is a $k$-dimensional matrix
\begin{equation}\label{E.37}
{{\bf{L}}^{k,q}} = \left| {{\bf{q}}\left( {{i_1},{i_2}, \ldots ,{i_k}} \right);1 \le {i_1},{i_2}, \ldots ,{i_k} \le q} \right|
\end{equation}
such that every row is a permutation of the set of natural numbers 1, 2, . . . , q. By
a row of  ${{\bf{L}}^{k,q}}$  we mean an $q$-tuple of elements ${{\bf{q}}\left( {{i_1},{i_2}, \ldots ,{i_k}} \right)}$ which have identical
coordinates $k-1$ at  places
\end{definition}
Using the definition of  Latin hypercube, we then have the  definition of  orthogonal Latin hypercubes with $N\left( q \right)$ tuples.
\begin{definition}
A $N\left( q \right)$-tuple of Latin $k$-dimensional cubes
\begin{equation}\label{E.38}
\left[ {{\bf{L}}_l^{k,q} = \left| {{{\bf{q}}_{l}}\left( {{i_1},{i_2}, \ldots ,{i_k}} \right);1 \le {i_1},{i_2}, \ldots ,{i_k} \le q} \right|} \right]
\end{equation}
of order $q$ for $l = 1,2, \ldots ,N\left( q \right)$ is called mutually  orthogonal, if whenever ${i_1},{i_2}, \ldots ,{i_k},i_1^{'},i_2^{'}, \ldots ,i_k^{'}\in \left\{ {1,2, \ldots ,q} \right\}$ are such that
\begin{equation}\label{E.39}
{{\bf{q}}_{l}}\left( {{i_1},{i_2}, \ldots ,{i_k}} \right) = {{\bf{q}}_{l}}\left( {i_1^{'},i_2^{'}, \ldots ,i_k^{'}} \right),for \,\,\, all \,1\le {l}\le N\left( q \right)
\end{equation}
then we must have  $i_{i}=i_{i}^{'}$ for all $i=1,\ldots, k$ . $N\left( q \right)$ represents the  maximum number of   orthogonal  Latin $k$-dimensional cubes of order $q$.
\end{definition}
The existence of orthogonal Latin $k$-dimensional cubes of order $q$ can be guaranteed by the following theorem
\begin{theorem}
For  $\forall k,  k\ge3$ and $\forall q,  q\ge3$, there exists a set of $k$ orthogonal Latin $k$-dimensional cubes of order $q$.
\end{theorem}
Then, we choose $k=3$  to construct  the orthogonal Latin cubes. The relationship between orthogonal Latin cubes and MDS code   is given by the following theorem.
\begin{theorem}
A $q$-nary  MDS code with $C = {q^3}$ and $d= r + 1$ is equivalent to a set of $r$ orthogonal Latin  cubes of order $q$ with $ q \ge 3, r=N\left( q \right)$.
\end{theorem}
\begin{proof}
Suppose we have a set of $r$ orthogonal  Latin three-dimensional cubes of order $q$. We first number the elements of three independent dimensions of  3-D cubes ( denoted by D1 D2 and D3, respectively)  using the same $q$ symbols from which the cubes are formed. Then we  construct a ${q^3}$ codewords supported code by using  D1 for the first position, D2  for the second, D3 for the third and the corresponding cube entries for the remaining $r$ positions. Two  codewords with three different digits at the first three positions cannot agree in  the last $r$ positions, since each of  codewords  designed from the orthogonal Latin cubes occurs exactly once. Furthermore, if two code words agree in either two of the first  three  positions, they can agree in none of the last $r$, since each of the $q$ symbols appears   exactly once in  any  row or column of  an arbitrary  2-dimensional slice of  Latin cubes. Similarly, if two code words agree in any one of the first three positions, they can also agree in none of the last $r$, since each of the paired $q$ symbols (totally $q^2$ symbols) on the 2-dimensional plane of a Latin cube appears exactly once in the set.
\end{proof}
\begin{figure}[!t]
\centering \includegraphics[width=1\linewidth]{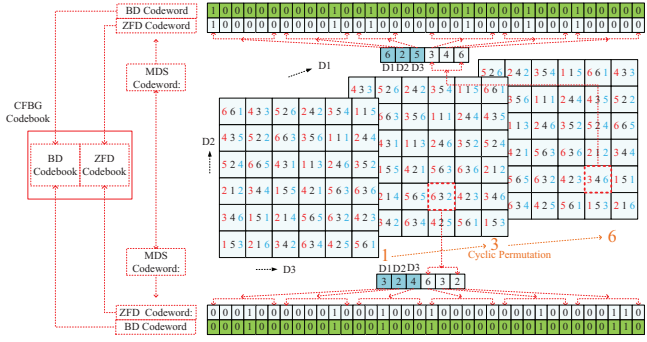}
\caption{Diagram of CFBG code construction process, including MDS code construction, ZFD code construction and BD code construction. In the beginning, MDS code is constructed through  3 orthogonal  6-order  Latin cubes.  There exist  a total of six  Latin cubes that are formulated by performing cyclic permutation for  three 6-order orthogonal Latin squares for six times. For each Latin square,  each  of columns undergoes cyclic permutation simultaneously. Note that a Latin square is a square array in which each row and each column consists of the same set of entries without repetition. Two $n \times n$ $l$-order Latin squares, denoted by ${\bf{X}} = \left[ {{x_{ij}}} \right]$ and ${\bf{Y}} =\left[ {{y_{ij}}} \right]$, $1\le{x_{ij}},{y_{ij}}\le l$, ${x_{ij}},{y_{ij}} \in {{\mathbb Z}^ + }$, are orthogonal iff $n^{2}$ pairs $\left( {{x_{ij}},{y_{ij}}} \right)$ are all different. Then ZFD code is constructed as the method mentioned in Theorem 4 and BD code is constructed by copying ZFD code but with different sum principle. }
\label{figure-CFBG}
\end{figure}
Finally, each  MDS codeword is formulated  by  searching  the three dimensions of Latin cubes  for  the first $k$ digits  and  then filling out the remnant $r$ positions with the searched value indicated by the three digits in the cube.  MDS codeword is further extended as the ZFD codeword by replacing each digit with the corresponding  $q$-digit binary vectors.
\subsubsection{Construction of CFBG Codebook}
After constructing ZFD code,  BD code can be obtained according to Definition 4. Thereafter, we  formulate  a  CFBG codebook  as follows:
 \begin{proposition}
 A   CFBG codebook  $\cal G$ is  a double-codeword (DCW) codebook defined by:
  \begin{equation}\label{E.40}
\cal G{\rm{ = }}\left\{ {\left[ {\begin{array}{*{20}{c}}
{{{\bf{g}}_{{\rm{C}},i}}}&{{{\bf{g}}_{{\rm{B}},i}}}
\end{array}} \right]\left| \begin{array}{l}
{{\bf{g}}_{C,i}} = {{\bf{g}}_{B,i}},{{\bf{g}}_{C,i}} \in {\cal C},{{\bf{g}}_{B,i}} \in {\cal B},\\
1 \le i \le C
\end{array} \right.} \right\}
  \end{equation}
  And the  superposition sum of DCWs in $\cal G$  is  defined as the two  independent  superposition sums where the sum between the first columns obeys  SP sum principle  while that between  the second columns  follows ASP sum principle.
 \end{proposition}
 An example of CFBG construction under  order $q=6$ is given in the Fig.~\ref{figure-CFBG}.   As previously introduced, the CFBG codebook is equally divided into two independent codebooks, respectively denoted by $\cal G_{\rm Bob}$ for Bob and $\cal G_{\rm Cha}$ for Charlie.  The  superposition sum set of codewords in $\cal G$ is defined by:
 \begin{definition}
The  superposition sum set ${\cal G}_{k-1}$ for $k = 2, 3$   is defined as the collection of all the superposition sums  of   DCWs in $\cal G$,  taken exactly $k$ at a time.
 \end{definition}
\subsection{ Codebook Performance}
In order to  measure  the  codebook performance, we develop the  concept of SEP, that is,  the  existence probability of duplicate codewords  among  the decomposed codewords.
\begin{theorem}
The SEP of Alice, that is, when Eva randomly selects one codeword in CFBG codebook for randomly-imitating attack,  is derived as:
\begin{equation}\label{E.41}
{\rm SEP} = \frac{{ 1}}{{{C}}} \propto {\left( {\frac{1}{{{N_{{\rm{Total}}}}{N_{\rm{T}}}}}} \right)^k}, k=3
\end{equation}
\end{theorem}
\begin{proof}
let us consider the number of  possible choices of codewords  for the  three independent nodes. As we know, Eva can attack  arbitrary node but Bob and Charlie only focus on their own codebook for distinguishing themselves from each other. Since  each codeword  is randomly selected,   the total number of the choices is equal to $\frac{C}{2} \times \frac{C}{2} \times C$ whereas the  duplicate codewords occur with  $\frac{C}{2} \times \frac{C}{2}$ possibilities. Therefore, we have ${\rm SEP} = \frac{{ 1}}{{{C}}}$.  Now we know $k=3$ and since $C = {q^k},\frac{{{N_{\rm{T}}}}}{{{N_B}}} = \frac{{q\left( {3k - 2} \right)}}{{{N_{{\rm{Total}}}}}}$, we can derive the SEP by ${\rm SEP} = {\left( {\frac{{7{N_{\rm{B}}}}}{{{N_{{\rm{Total}}}}{N_{\rm{T}}}}}} \right)^3}$. As shown in Eq.~(\ref{E.31}), the control variable $N_{\rm B}$ is artificially configured and usually fixed.  In this sense, we prove the theorem.
\end{proof}
To simulate  SEP,  $\rm SEP_{dB}$  is defined as follows:
\begin{equation}\label{E.42}
{\rm{SE}}{{\rm{P}}_{{\rm{dB}}}} = 10{\log _{10}}{\rm{SEP}}
\end{equation}
We configure $N_{\rm{B}}$ to be 100.  Note that at most $N_{\rm T}=100$ antennas are  supported in this example.  The number of antennas is however not constrained,  if needed.
Fig.~\ref{figure2}(b) shows the value of $\rm SEP_{dB}$ versus  $N_{\rm T}$  and $N_{\rm Total}$.  As we can see, $\rm SEP_{dB}$ decreases with the increase of  $N_{\rm T}$ and $N_{\rm Total}$.  This
accords  with what is shown in Eq.~(\ref{E.41}). Specifically, the $\rm SEP_{dB}$  reaches  -56 under  $N_{\rm T}=100$ and $N_{\rm Total}=500$.  Fig.~\ref{figure2}(c) demonstrates the tradeoff  between  $N_{\rm T}$ and $N_{\rm Total}$ given the value of $\rm SEP_{dB}$. Obviously, the number of subcarriers occupied for guaranteeing a desirable $\rm SEP_{dB}$ is reduced with the increase of the number $N_{\rm T}$ of antennas.  This reduction increases  with the decrease  of $\rm SEP_{dB}$.
\section{Key Technique III: Pilot Encoding $\&$ Decoding}
\label{PED}
We introduce the pilot encoding and  decoding process through the formulated codebook. An example can be shown in Fig.~\ref{Conveying}. Finally,  we will identify the  unsolved issues.
\begin{algorithm}[!t] 
\caption{ BDCD Algorithm} 
\label{alg1} 
\begin{algorithmic}[1] 
\REQUIRE Observation codeword $\bf g={\left[ {\begin{array}{*{20}{c}}
{{{\bf{g}}_{{\rm{C}}}}}&{{{\bf{g}}_{{\rm{B}}}}}
\end{array}} \right]}$,  ${\cal G}_{\rm Bob}$, ${\cal G}_{\rm Cha}$.
\ENSURE Original codewords: ${\bf g}_{{\rm C},1}$, ${\bf g}_{{\rm C}, 2}$, and ${\bf g}_{{\rm C}, 3}$. 
\FOR {$\bf g$}
\STATE In the set ${\cal G} \cup {\cal G}_{1} \cup {\cal G}_{2}$,  find the codeword equal to ${{{\bf{g}}_{{\rm{C}}}}}$.
\IF {True}
\IF {Each digit of ${{{{g}}_{{\rm{C}}}}}$ is equal to one }\STATE{Calculate the element  ${{{{g}}_{{{\rm{B}},i}}}}$ of ${{{{g}}_{{\rm{B}}}}}$}
\IF {There exist $i$ such that  ${{{{g}}_{{{\rm{B}},i}}}} = 1$}
\STATE Based on Property 1, reduce  each digit of ${{{\bf{g}}_{{\rm{C}}}}}$  at the same digit positions by one. Search the resulted codeword in ${\cal G} \cup {\cal G}_{1} \cup {\cal G}_{2}$ again.
\IF{True}
\STATE Indicate  jamming attack. Output  ${\bf g}_{{\rm C},1}$, ${\bf g}_{{\rm C}, 2}$.
\ELSE
\STATE Indicate  no attack. Output ${\bf g}_{{\rm C},1}$, ${\bf g}_{{\rm C}, 2}$, ${\bf g}_{{\rm C}, 3}$.
\ENDIF
\ELSE
\STATE Indicate  jamming attack. Output  ${\bf g}_{{\rm C},1}$, ${\bf g}_{{\rm C}, 2}$.
\ENDIF
\ELSE
\STATE Interpret $\bf g$ as  original codewords using Principle~\ref{Proposition3}\&~\ref{Proposition4}.
\IF {Number of  decomposed codewords is  three}
\STATE Output ${\bf g}_{{\rm C},1}$, ${\bf g}_{{\rm C},2}$, ${\bf g}_{{\rm C},3}$.
\ENDIF
\IF{Number of  decomposed codewords is  two}
\STATE Eatimate the codeword ${\left[ {\begin{array}{*{20}{c}}
{{{\overline {\bf{g}}}_{{\rm{C}},i,1}}}&{{{\overline {\bf{g}}}_{{\rm{B}},i,1}}}
\end{array}} \right]}$ from  ${\cal G}_{\rm Bob}$  and ${\left[ {\begin{array}{*{20}{c}}
{{{\overline {\bf{g}}}_{{\rm{C}},i,2}}}&{{{\overline {\bf{g}}}_{{\rm{B}},i,2}}}
\end{array}} \right]}$ from  ${\cal G}_{\rm Cha}$ . Detect each digit of  ${{{\overline {\bf{g}}}_{{\rm{B}},i,1}}}$ and ${{{\overline {\bf{g}}}_{{\rm{B}},i,2}}}$ and calculate the sum of digits respectively as $s_1$ and $s_2$.
\IF {$s_1=s_2$}
\STATE Indicate no error and Output   ${{{\overline {\bf{g}}}_{{\rm{C}},i,1}}}$, ${{{\overline {\bf{g}}}_{{\rm{C}},i,2}}}$.
\ELSE
\STATE Indicate separation error.
\ENDIF
\ENDIF
\ENDIF
\ELSE
\STATE   Indicate jamming attack.  Reduce each digit of $\bf g$ by one and output  ${\bf g}_{{\rm C},1}$, ${\bf g}_{{\rm C},2}$ by searching  ${\cal G} \cup {\cal G}_{1} \cup {\cal G}_{2}$.
\ENDIF
\ENDFOR
\end{algorithmic}
\end{algorithm}
\begin{figure*}
\centering \includegraphics[width=1.00\linewidth]{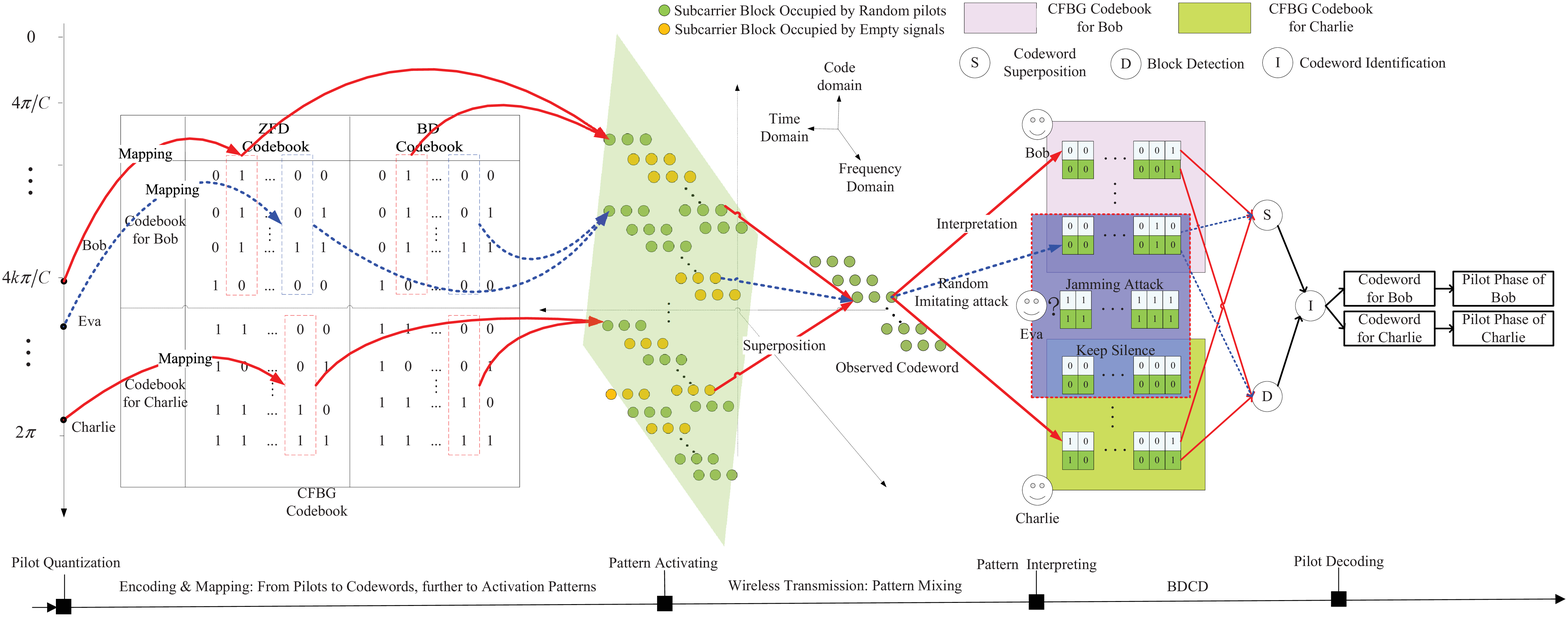}
\caption{Example of pilot encoding and decoding  process by using CFBG DCW codebook under randomly-imitating attack.}
\label{Conveying}
\end{figure*}

\subsubsection{ Pilot Quantization and Encoding} The common phase interval $\left[ {0,2\pi } \right)$ is equally quantized into $C$ reference values. Then an one-to-one mapping is formulated between each phase value and a corresponding  codeword. Every time Bob (or Charlie) has utilized a random pilot, such as ${\theta  _{k_0}}$ ( or ${\beta_{k_0}}$) at one symbol time (i.e., $k_0$) , it compares the phase with reference values, selects the reference value  ${\overline \theta  _{k_0}}$ (or ${\overline \beta_{k_0}}$) closest to the utilized phase and maps the value into a codeword. Finally, two codebooks denoted by  $\cal G_{\rm Bob}$ and  $\cal G_{\rm Cha}$ are respectively allocated  for Bob and Charlie. Bob selects ${\bf g}_{{\rm C},1}$,  Charlie selects ${\bf g}_{{\rm C}, 2}$ and Eva, if existing, selects  ${\bf g}_{{\rm C}, 3}$.

\subsubsection{Block Pattern Activating}
Bob (or Charlie) maps ${\bf g}_{{\rm C},1}$ (or ${\bf g}_{{\rm C},2}$) into the subcarrier block activation patterns. The principle is that Bob (or Charlie)  transmits signals on the $i$-th subcarrier block if the $i$-th ($1 \le i \le  B$) row element of ${{\bf{g}}_{{\rm C},1}}$ (${\bf g}_{{\rm C},2}$) is equal to 1, otherwise Bob (or Charlie) keeps silent on this subcarrier block.

\subsubsection{Block Pattern Interpreting}
 Due to the overlapping  of  activation patterns from three nodes,  the  interpretation of this pattern into original  codewords  requires  the combination of block detection technique and  CFBG  code. The detailed algorithm is summarized  in Algorithm 1.
\subsubsection{Pilot Decoding Based on  Interpreted Codewords}
Alice identifies  the interpreted codewords as quantized pilot phases, i.e., ${\overline \theta  _{k_0}}$  and  ${\overline \beta_{k_0}}$  and finally recovers the pilot signals.

 For this  process, what  is certain is that three types of  attack can be identified  perfectly. As to the codeword identification, we will encounter three situations: 1) When Eva  keeps silence,  two  pilots from legitimate nodes can be separated and identified; 2) Under  jamming attack,  partial pilots, i.e. belonging to Bob and Charlie,  can be separated and identified, which is enough for the following channel estimation; 3) Under randomly-imitating attack,   we have the following problem:
\begin{problem}
Two interpreted codewords within the same codebook, though separated from each other,  cannot be identified.
\end{problem}
Thanks to CFBG codebook, what we can achieve until now is conveying and separating pilots perfectly while identifying pilots with a certain level of errors. We should note that high-resolution  codeword separation  logically acts as a necessary step towards high-resolution pilot identification and  provides a basis of identification   enhancement  in the following section.
\section{Key Technique IV: Joint Channel Estimation and Identification}
\label{JCEI}
In order to solve above issue and further achieve  the critical and final goal, i. e., channel  acquisition,  we  focus on the channel estimation process.  Generally, when there is no attack, a well-known LS estimator is enough for channel estimation.  Therefore,  we in this section turn to  the attack environment.   We aim to: 1) design  the high-precision channel estimator; 2) design a pilot (or channel) identification enhancement mechanism for  randomly-imitating attack.
\subsection{Signal Representation for Channel Estimation}
We begin our discussion by  stacking  the signals received within the first three  OFDM symbol time as
\begin{equation}\label{E.43}
{\overline{\bf{Y}}} = {\bf{XH}} + {\bf{N}}
 \end{equation}
Here, we have  ${\overline{\bf{Y}}} = {\left[ {\begin{array}{*{20}{c}}
{{{\bf{y}}^{\rm{T}}}\left[ {{k_0}} \right]}&{{{\bf{y}}^{\rm{T}}}\left[ {{k_1}} \right]}&{{{\bf{y}}^{\rm{T}}}\left[ {{k_2}} \right]}
\end{array}} \right]^{\rm{T}}}$, ${\bf{X}} = \left[ {\begin{array}{*{20}{c}}
{{{\bf{x}}_{\rm{B}}}}&{{{\bf{x}}_{\rm{C}}}}&{{{\bf{x}}_{\rm{E}}}}
\end{array}} \right]$ and ${\bf{N}} = \left[ {\begin{array}{*{20}{c}}
{{{\bf{w}}^{\rm{T}}}\left[ {{k_0}} \right]}&{{{\bf{w}}^{\rm{T}}}\left[ {{k_1}} \right]}&{{{\bf{w}}^{\rm{T}}}\left[ {{k_2}} \right]}
\end{array}} \right]$. There exist ${{\bf{x}}_{{j}}} = {\left[ {{x_{{j}}}{{\left[ {{k_i}} \right]}_{0 \le i \le 2}}} \right]^{\rm{T}}} \in {{\mathbb C}^{3 \times 1}}$, $j \in \left\{ {{\rm{B,C,E}}} \right\}$ and ${\bf{H}}{\rm{ = }}{\left[ {\begin{array}{*{20}{c}}
{{\bf{h}}_{\rm{B}}^{\rm{T}}}&{{\bf{h}}_{\rm{C}}^{\rm{T}}}&{{\bf{h}}_{\rm{E}}^{\rm{T}}}
\end{array}} \right]^{\rm{T}}}$. We define ${{\bf{h}}_{{j}}} = {{\bf{g}}_{{j}}}\left( {{{\bf{I}}_{{N_{\rm{T}}}}} \otimes {\bf{F}}_{\rm{L}}^{\rm{T}}} \right), j \in \left\{ {{\rm{B,C,E}}} \right\}$
where ${{\bf{g}}_{{j}}} = \left[ {\begin{array}{*{20}{c}}
{{{\left( {{\bf{h}}_{{j}}^1} \right)}^{\rm{T}}}}&{, \ldots ,}&{{{\left( {{\bf{h}}_{{j}}^{{N_{\rm{T}}}}} \right)}^{\rm{T}}}}
\end{array}} \right]\in {{\mathbb C}^{1 \times {N_{\rm{T}}}L}}, j \in \left\{ {{\rm{B,C,E}}} \right\}$.  It is easily to verify
${{\bf{h}}_{{j}}}{\bf{h}}_{{j}}^{\rm{H}} = N{{\bf{g}}_{{j}}}{\bf{g}}_{{j}}^{\rm{H}},  j \in \left\{ {{\rm{B,C,E}}} \right\}$.
 Then we define ${\bf{g}}_{{j}}^{\rm{H}} = {\left( {{{\bf{R}}_{{j}}} \otimes {{\bf{I}}_{L}}} \right)^{\frac{1}{2}}}\widetilde {\bf{g}}_{{j}}^{\rm{H}},  j \in \left\{ {{\rm{B,C,E}}} \right\}$,
  where each ${\widetilde {\bf{g}}_{{j}}} \sim {\cal C}{\cal N}\left( {0,{{\bf{I}}_{{N_{\rm{T}}}L}}} \right)$ for  $j \in \left\{ {{\rm{B,C,E}}} \right\}$ is a $1 \times N_{\rm T}L$ vector.
  Finally we derive  the relationship between FS and CIR as follows
\begin{equation}\label{E.44}
{{\bf{h}}_{{j}}} = {\widetilde {\bf{g}}_{{j}}}\left( {{\bf{R}}_{{j}}^{\frac{1}{2}} \otimes {\bf{F}}_{\rm{L}}^{\rm{T}}} \right),  j \in \left\{ {{\rm{B,C,E}}} \right\}
\end{equation}

From the \emph{Lemma}
B.26 in ~\cite{Hoydis},  we  derive the  asymptotic approximation for FS channels $ j \in \left\{ {{\rm{B,C,E}}} \right\}$ by
$\frac{1}{{{N_{\rm{T}}}N}}{{\bf{h}}_{{j}}}{\bf{h}}_{{j}}^{\rm{H}}  \xlongrightarrow[{N_{\rm{T}}} \to \infty]{ \rm{a.s.}} \frac{1}{{{N_{\rm{T}}}}}{\rm{Tr}}\left( {{{\bf{R}}_{{j}}} \otimes {{\bf{I}}_{{L}}}} \right) = \frac{L}{{{N_{\rm{T}}}}}{\rm{Tr}}\left( {{{\bf{R}}_{{j}}}} \right)$.
Similarly, we can obtain the following asymptotic results:
$\frac{1}{{{N_{\rm{T}}}N}}{{\bf{h}}_{{j}}}{\bf{h}}_{{l}}^{\rm{H}} \xlongrightarrow[{N_{\rm{T}}} \to \infty]{ \rm{a.s.}} 0, \forall j \ne l,  j,l \in \left\{ {{\rm{B,C,E}}} \right\}$,
$\frac{1}{{{N_{\rm{T}}}N}}{\bf{w}}\left[ {{k_i}} \right]{{\bf{w}}^{\rm{H}}}\left[ {{k_j}} \right] \xlongrightarrow[{N_{\rm{T}}} \to \infty]{ \rm{a.s.}} {\sigma ^{\rm{2}}}, \forall i = j$, and
$\frac{1}{{{N_{\rm{T}}}N}}{\bf{w}}\left[ {{k_i}} \right]{{\bf{w}}^{\rm{H}}}\left[ {{k_j}} \right] \xlongrightarrow[{N_{\rm{T}}} \to \infty]{ \rm{a.s.}} 0, \forall i \ne j$.
We consider  the covariance matrix  defined by ${{\bf C}_{\overline{\bf{Y}}}} = \frac{1}{{{N_{\rm{T}}N}}}{\overline{\bf{Y}}}\,{{\overline{\bf{Y}}}}^{\rm{H}}$ satisfying:
\begin{equation}\label{E.49}
\mathop {{\bf C}_{\overline{\bf{Y}}}}\xlongrightarrow[{N_{\rm{T}}} \to \infty]{ \rm{a.s.}}\frac{L}{{{N_{\rm{T}}}}}{\bf{XR}}{{\bf{X}}^{\rm{H}}} + {\sigma ^2}{{\bf{I}}_2}
\end{equation}
where ${\bf{R}} = {\rm{diag}}\left\{ {\begin{array}{*{20}{c}}
{{\rm{Tr}}\left( {{{\bf{R}}_{\rm{B}}}} \right)}&{{\rm{Tr}}\left( {{{\bf{R}}_{\rm{C}}}} \right)}&{{\rm{Tr}}\left( {{{\bf{R}}_{\rm{E}}}} \right)}
\end{array}} \right\}$
\subsection{Design of Channel Estimator}
Using CFBG codebook and demapping operation, Alice  derives two separated  pilot phases  ${\overline \theta }_{{k_0}}$ and ${\overline \beta }_{{k_0}}$, and thus deduce the first two columns of $\bf{X}$, denoted by ${\bf{\overline x}}_{{i}}$ for the $i$-th column, expressed by:
$
{{{\bf{\overline x}}}_1} = {\left[ {\begin{array}{*{20}{c}}
{\sqrt {{\rho _{\rm{B}}}} {e^{j{{\overline \theta }_{{k_0}}}}}}&{\sqrt {{\rho _{\rm{B}}}} {e^{j\left( {{{\overline \theta }_{{k_0}}} + \theta } \right)}}}&{\sqrt {{\rho _{\rm{B}}}} {e^{j\left( {{{\overline \theta }_{{k_0}}} + 2\theta } \right)}}}
\end{array}} \right]^{\rm{T}}}$,
${{{\bf{\overline x}}}_2} = {\left[ {\begin{array}{*{20}{c}}
{\sqrt {{\rho _{\rm{C}}}} {e^{j{{\overline \beta }_{{k_0}}}}}}&{\sqrt {{\rho _{\rm{C}}}} {e^{j\left( {{{\overline \beta }_{{k_0}}} + \beta } \right)}}}&{\sqrt {{\rho _{\rm{C}}}} {e^{j\left( {{{\overline \beta }_{{k_0}}} + 2\beta } \right)}}}
\end{array}} \right]^{\rm{T}}}$
 The design principle   is to derive FS and CIR based on $\mathop {{\bf C}_{\overline{\bf{Y}}}}$ and  ${\bf{\overline x}}_{{i}}$.
\begin{figure*}
\centerline{
  ~\includegraphics[width=2.50in]{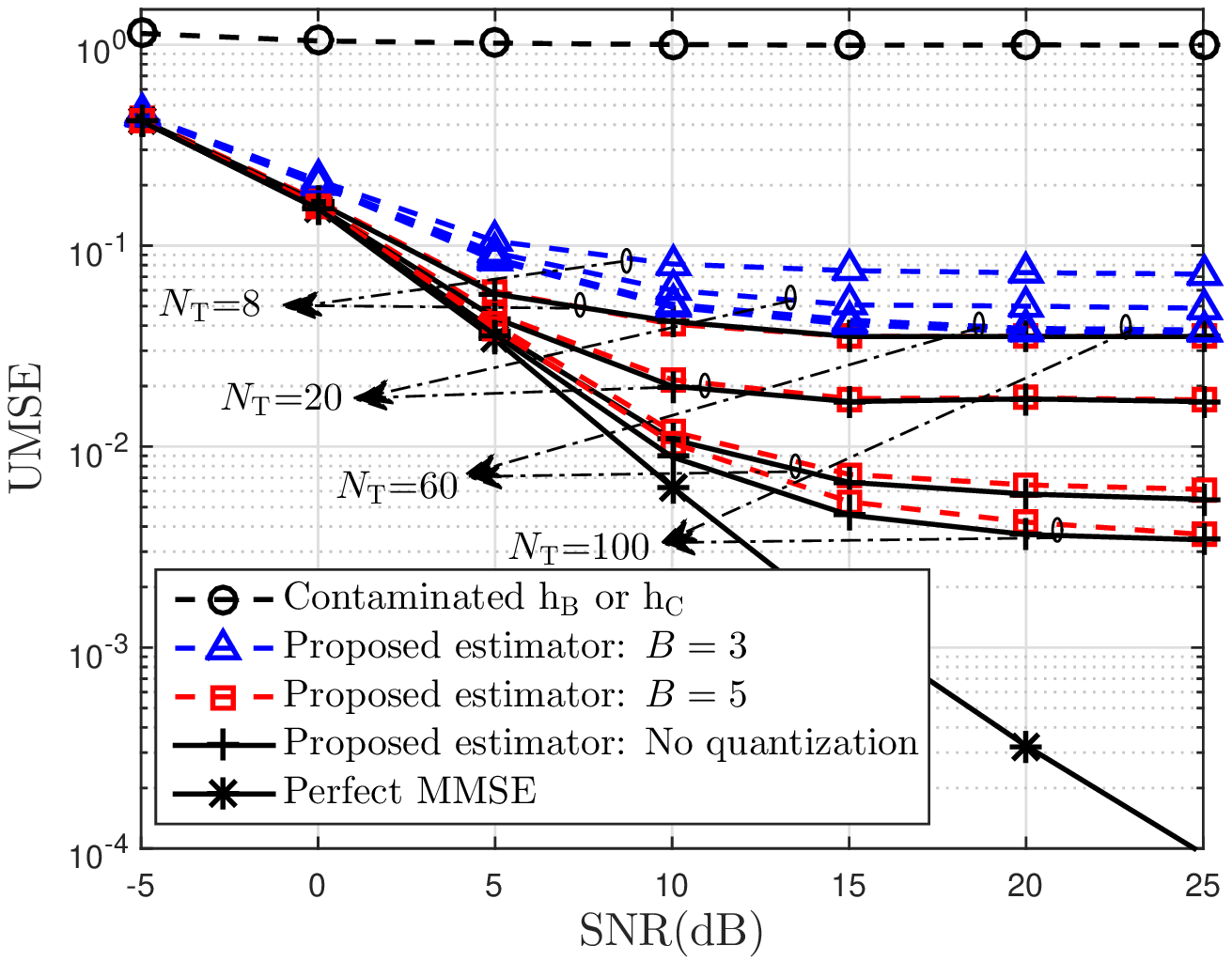} \hspace{-10pt} \includegraphics[width=2.50in]{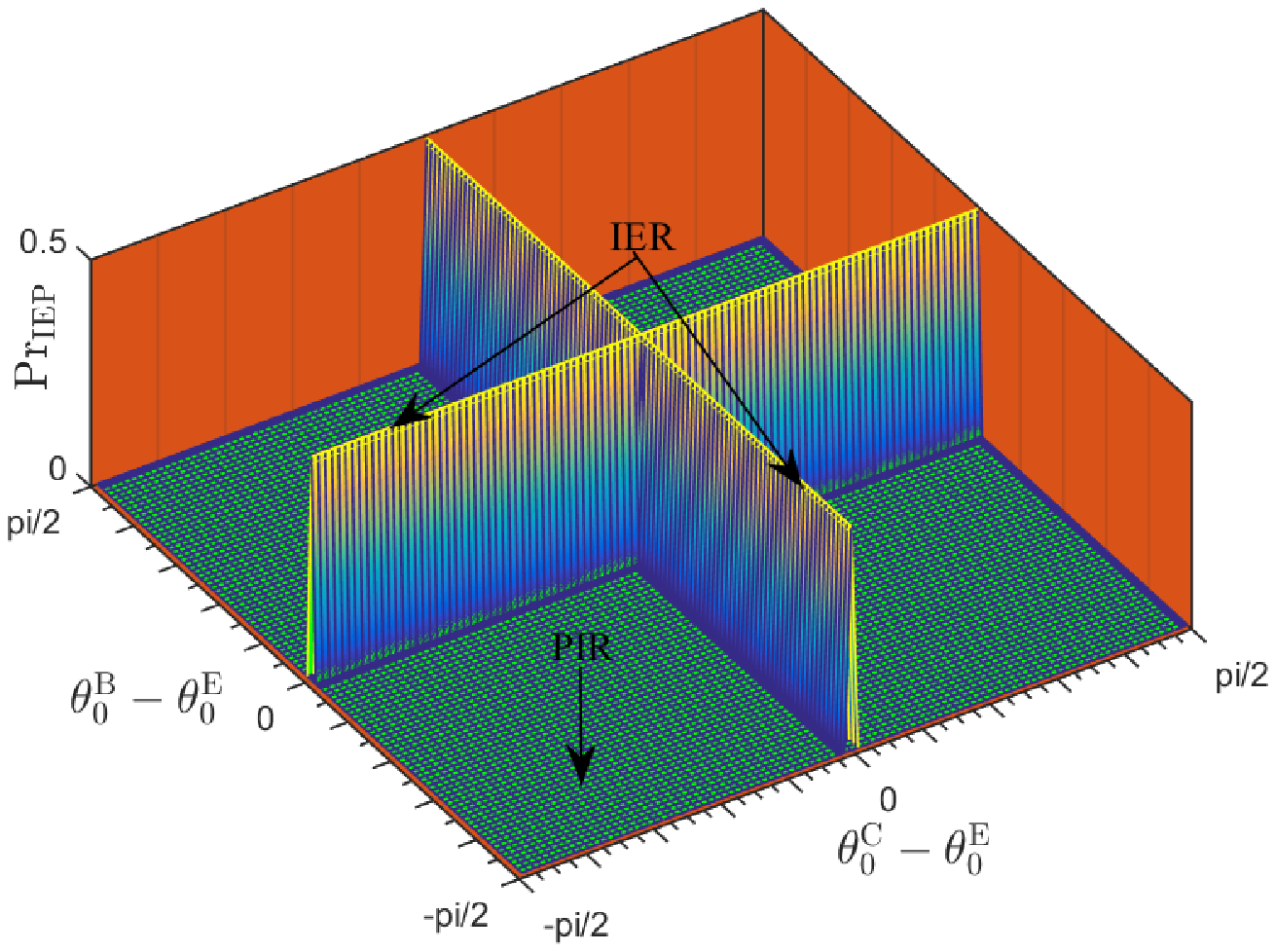}\hspace{-10pt}  \includegraphics[width=2.50in]{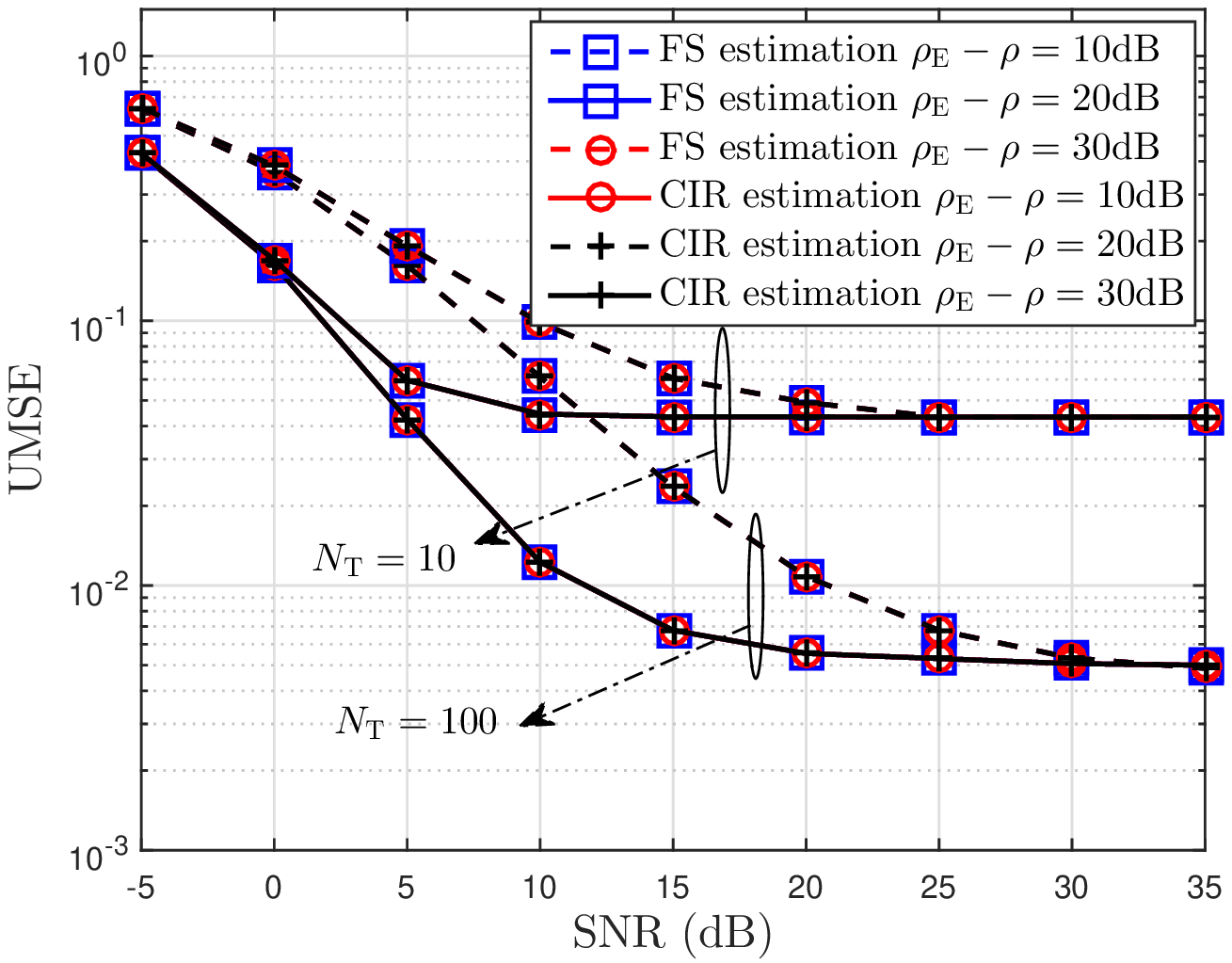}}
  \centerline{\hspace{10pt}(a)\hspace{160pt}(b)\hspace{160pt}(c)}\vspace{-5pt}
  \caption{Performance evaluation of channel estimation and identification;  (a)UMSE  versus SNR  and  $B$; (b)3D plot of  IEP versus the mean  AoA separations  under $N_{\rm T}= 64$; (c) UMSE of FS and CIR estimation versus SNR under various  power difference.}\vspace{10pt}
  \label{figure3}
\end{figure*}

From  Eq.~(\ref{E.49}), we can derive  the MMSE semi-blind estimators for FS channels as
${{\bf{W}}_{\rm{F,B}}} = \sqrt {\frac{{L{\rm{Tr}}\left( {{{\bf{R}}_{\rm{B}}}} \right)}}{{{N_{\rm{T}}}}}} {\bf{\overline x}}_{\rm{1}}^{\rm{H}}{\bf C}_{{\overline{\bf{Y}}}}^{ - 1}, {{\bf{W}}_{\rm{F,C}}} = \sqrt {\frac{{L{\rm{Tr}}\left( {{{\bf{R}}_{\rm{C}}}} \right)}}{{{N_{\rm{T}}}}}} {\bf{\overline x}}_{\rm{2}}^{\rm{H}}{\bf C}_{{\overline{\bf{Y}}}}^{ - 1}$.
The estimated versions of FS channels are  respectively  derived  by ${{\widehat{\bf{h}}}_{\rm{B}}}{\rm{ = }}{{\bf{W}}_{\rm{F,B}}}{\overline{\bf{Y}}}$ and ${{\widehat{\bf{h}}}_{\rm{C}}}{\rm{ = }}{{\bf{W}}_{\rm{F,C}}}{\overline{\bf{Y}}}$.
In the following,  we first eliminate the influence of FFT weight by multiplying  ${{{\widehat {\bf{h}}}_{{j}}}}, j \in \left\{ {{\rm{B,C}}} \right\}$    by a right-weighting matrix ${{\bf{I}}_{{N_{\rm{T}}}}} \otimes \left\{ {{{\left( {{\bf{F}}_{\rm{L}}^{\rm{T}}} \right)}^{\rm{H}}}{{\left( {{\bf{F}}_{\rm{L}}^{\rm{T}}{\bf{F}}_{\rm{L}}^{\rm{*}}} \right)}^{ - 1}}} \right\}$. The result is then multiplied by ${\bf{R}}_{{j}}^{ - \frac{1}{2}} \otimes {{\bf{I}}_L}, j \in \left\{ {{\rm{B,C}}} \right\}$ to eliminate the influence of spatial correlation.  Finally,  the CIR estimations are derived  as
  \begin{equation}\label{E.52}
{\widehat {\bf{g}}_{{j}}}{\rm{ = }}{\widehat {\bf{h}}_{{j}}}\left\{ {{\bf{R}}_{{j}}^{ - \frac{1}{2}} \otimes \left\{ {{{\left( {{\bf{F}}_{\rm{L}}^{\rm{T}}} \right)}^{\rm{H}}}{{\left( {{\bf{F}}_{\rm{L}}^{\rm{T}}{\bf{F}}_{\rm{L}}^{\rm{*}}} \right)}^{ - 1}}} \right\}} \right\}, j \in \left\{ {{\rm{B,C}}} \right\}
   \end{equation}
\subsection{Identification Enhancement}
  For randomly-imitating attack,  CFBG codebook provides three separated  pilots. Three estimated channels can thus be derived using  the above same principle.  In this context, \emph{channel identification is equivalent to  pilot identification since each estimator only relies on one corresponding pilot signal}.  For simplicity, we denote Eva's  pilot signal recovered by:
\begin{equation}\label{E.53}
{{{\bf{\overline x}}}_3} = {\left[ {\begin{array}{*{20}{c}}
{\sqrt {{\rho _{\rm{E}}}} {e^{j{{\overline \varphi }_{{k_0}}}}}}&{\sqrt {{\rho _{\rm{E}}}} {e^{j\left( {{{\overline \varphi }_{{k_0}}} + \varphi } \right)}}}&{\sqrt {{\rho _{\rm{E}}}} {e^{j\left( {{{\overline \varphi }_{{k_0}}} + 2\varphi } \right)}}}
\end{array}} \right]^{\rm{T}}}
\end{equation}
where ${\overline \varphi }_{{k_0}}$ is the recovered pilot phase indicated by the confusing codeword.  Its estimation version of CIR satisfies
\begin{equation}\label{E.54}
{\widehat {\bf{g}}_{\rm{E}}}{\rm{ = }}{\widehat {\bf{h}}_{\rm{E}}}\left\{ {{\bf{R}}_{\rm{E}}^{ - \frac{1}{2}} \otimes \left\{ {{{\left( {{\bf{F}}_{\rm{L}}^{\rm{T}}} \right)}^{\rm{H}}}{{\left( {{\bf{F}}_{\rm{L}}^{\rm{T}}{\bf{F}}_{\rm{L}}^{\rm{*}}} \right)}^{ - 1}}} \right\}} \right\}
\end{equation}
 Since CFBG  codebook guarantees that Alice can identify  which node is under randomly-imitating attack, we turn to  design of the identification mechanism for those channels under  attack. Take Bob for example,  we aim to identify ${\widehat {\bf{g}}_{\rm{B}}}$ and ${\widehat {\bf{g}}_{\rm{E}}}$ by applying  maximum-likelihood detection  (MLD) and the available spatial correlation. The operation for Charlie, if being misguided, has the same methodology.    Note that the probability distribution of  ${{\bf{g}}_{\rm{B}}}$  is available at Alice and  given in~\cite{Goodman} by
${p_{{{\bf{g}}_{\rm{B}}}}}\left( {\bf{r}} \right) = \frac{{\exp \left[ { - \frac{1}{2}{{\bf{r}}}\left( {{\bf{R}}_{\rm{B}}^{ - 1} \otimes {{\bf{I}}_L}} \right){\bf{r}}^{\rm{H}}} \right]}}{{{{\left( {2\pi } \right)}^{{{{N_{\rm{T}}}L{\rm{ }}} \mathord{\left/
 {\vphantom {{{N_{\rm{T}}}L{\rm{ }}} 2}} \right.
 \kern-\nulldelimiterspace} 2}}}{{\left| {{{\bf{R}}_{\rm{B}}} \otimes {{\bf{I}}_L}} \right|}^{{1 \mathord{\left/
 {\vphantom {1 2}} \right.
 \kern-\nulldelimiterspace} 2}}}}}$.
After deriving  the conditional density ${p_{{{\bf{g}}_{{ \rm{B}}}}}}\left( {{\bf{r}}\left| {{{\bf{R}}_{{\rm{B}}}}} \right.} \right)$  based on ${{\bf{R}}_{{\rm{B}}}}$,  we construct the  likelihood function  ${\rm{In}}\left( {{p_{{{\bf{g}}_{{\rm{B}}}}}}\left( {{\bf{r}}\left| {{{\bf{R}}_{{\rm{B}}}}} \right.} \right)} \right)$ and  formulate the identification problem as $\widehat {\bf{h}}{\rm{ = }}\mathop {\arg \max }\limits_{{\bf{r}} = {{\widehat {\bf{g}}}_{\rm{B}}},{{\widehat {\bf{g}}}_{\rm{E}}}} \left\{ {{\rm{In}}\left( {{p_{{{\bf{g}}_{\rm{B}}}}}\left( {{\bf{r}}\left| {{{\bf{R}}_{\rm{B}}}} \right.} \right)} \right)} \right\}$
  which is then equivalently  transformed into
  \begin{equation}\label{E.57}
\widehat {\bf{h}}{\rm{ = }}\mathop {\arg \min }\limits_{{\bf{r}} = {{\widehat {\bf{g}}}_{\rm{B}}},{{\widehat {\bf{g}}}_{\rm{E}}}} \left\{ {{\bf{r}}\left( {{\bf{R}}_{\rm{B}}^{ - 1} \otimes {{\bf{I}}_L}} \right){{\bf{r}}^{\rm{H}}}} \right\}
   \end{equation}
   Then the  IEP   is finally defined by:
   \begin{equation}\label{E.58}
 {\rm Pr_{{IEP}}}= \Pr \left\{ {{{\widehat {\bf{g}}}_{\rm{B}}}\left( {{\bf{R}}_{\rm{B}}^{ - 1} \otimes {{\bf{I}}_L}} \right)\widehat {\bf{g}}_{\rm{B}}^{\rm{H}} > {{\widehat {\bf{g}}}_{\rm{E}}}\left( {{\bf{R}}_{\rm{B}}^{ - 1} \otimes {{\bf{I}}_L}} \right)\widehat {\bf{g}}_{\rm{E}}^{\rm{H}}} \right\}
   \end{equation}
 \begin{proposition}
 The asymptotic IEP  when $N_{\rm T}\to \infty$   is given by:
 \begin{equation}\label{E.59}
 {\rm Pr^{\infty}_{{IEP}}} =\Pr \left\{ A{{\rm{Tr}}\left( {{\bf{R}}_{\rm{B}}^{ - 2}} \right)>B{\rm{ Tr}}\left( {{\bf{R}}_{\rm{E}}^{ - 1}{\bf{R}}_{\rm{B}}^{ - 1}} \right)} \right\}
 \end{equation}
where $A = 1 - {{L{\rm{Tr}}\left( {{{\bf{R}}_{\rm{B}}}} \right)\overline {\bf{x}} _{\rm{1}}^{\rm{H}}{\bf{C}}{{\overline {\bf{x}} }_1}} \mathord{\left/
 {\vphantom {{L{\rm{Tr}}\left( {{{\bf{R}}_{\rm{B}}}} \right)\overline {\bf{x}} _{\rm{1}}^{\rm{H}}{\bf{C}}{{\overline {\bf{x}} }_1}} {{N_{\rm{T}}}}}} \right.
 \kern-\nulldelimiterspace} {{N_{\rm{T}}}}},B = 1 - {{L{\rm{Tr}}\left( {{{\bf{R}}_{\rm{C}}}} \right)\overline {\bf{x}} _{{3}}^{\rm{H}}{\bf{C}}{{\overline {\bf{x}} }_3}} \mathord{\left/
 {\vphantom {{L{\rm{Tr}}\left( {{{\bf{R}}_{\rm{C}}}} \right)\overline {\bf{x}} _{{3}}^{\rm{H}}{\bf{C}}{{\overline {\bf{x}} }_3}} {{N_{\rm{T}}}}}} \right.
 \kern-\nulldelimiterspace} {{N_{\rm{T}}}}}$
   and ${\bf{C}} = {\left( {\frac{L}{{{N_{\rm{T}}}}}{\bf{XR}}{{\bf{X}}^{\rm{H}}} + {\sigma ^2}{{\bf{I}}_2}} \right)^{ - 1}}$.
 \end{proposition}
\begin{proof}
There exists ${\widehat {\bf{g}}_{\rm{B}}}\left( {{\bf{R}}_{\rm{B}}^{ - 1} \otimes {{\bf{I}}_L}} \right)\widehat {\bf{g}}_{\rm{B}}^{\rm{H}} = {\widehat {\bf{h}}_{\rm{B}}}\left\{ {{\bf{R}}_{\rm{B}}^{ - 2} \otimes \left\{ {{{\left( {{\bf{F}}_{\rm{L}}^{\rm{T}}} \right)}^{\rm{H}}}{{\left( {{\bf{F}}_{\rm{L}}^{\rm{T}}{\bf{F}}_{\rm{L}}^{\rm{*}}} \right)}^{ - 2}}{\bf{F}}_{\rm{L}}^{\rm{T}}} \right\}} \right\}\widehat {\bf{h}}_{\rm{B}}^{\rm{H}}$.   It is easily  shown that  ${\widehat {\bf{h}}_{\rm{B}}} = {{\bf{h}}_{\rm{B}}} - \sqrt A {\bf{h}}$
 where there exists  $A = 1 - {{L{\rm{Tr}}\left( {{{\bf{R}}_{\rm{B}}}} \right){\bf{x}}_{\rm{1}}^{\rm{H}}{\mathbb E} \left\{ {{\bf C}_{{\overline{\bf{Y}}}}^{ - 1}} \right\}{{\bf{x}}_1}} \mathord{\left/
 {\vphantom {{L{\rm{Tr}}\left( {{{\bf{R}}_{\rm{B}}}} \right){\bf{x}}_{\rm{1}}^{\rm{H}}{{\mathbb E} \left\{ {{\bf C}_{{\overline{\bf{Y}}}}^{ - 1}} \right\}}{{\bf{x}}_1}} {{N_{\rm{T}}}}}} \right.
 \kern-\nulldelimiterspace} {{N_{\rm{T}}}}}$. $ \sqrt A{\bf{h}}$ is the estimation error that is uncorrelated with ${{\bf{h}}_{\rm{B}}}$. The entries of ${\bf{h}}$ are i.i.d zero-mean complex Gaussian with unity variance. Therefore, we have
${\widehat {\bf{g}}_{\rm{B}}}\left( {{\bf{R}}_{\rm{B}}^{ - 1} \otimes {{\bf{I}}_L}} \right)\widehat {\bf{g}}_{\rm{B}}^{\rm{H}} = {{\bf{h}}_{\rm{B}}}\overline {\bf{R}} {\bf{h}}_{\rm{B}}^{\rm{H}} - 2A{\bf{h}}\overline {\bf{R}} {\bf{h}}_{\rm{B}}^{\rm{H}} + A{\bf{h}}\overline {\bf{R}} {{\bf{h}}^{\rm{H}}}$
 where $\overline {\bf{R}}  = {\bf{R}}_{\rm{B}}^{ - 2} \otimes \left\{ {{{\left( {{\bf{F}}_{\rm{L}}^{\rm{T}}} \right)}^{\rm{H}}}{{\left( {{\bf{F}}_{\rm{L}}^{\rm{T}}{\bf{F}}_{\rm{L}}^{\rm{*}}} \right)}^{ - 2}}{\bf{F}}_{\rm{L}}^{\rm{T}}} \right\}$. After simplifying each term using asymptotic approximation, we  derive
 ${\widehat {\bf{g}}_{\rm{B}}}\left( {{\bf{R}}_{\rm{B}}^{ - 1} \otimes {{\bf{I}}_L}} \right)\widehat {\bf{g}}_{\rm{B}}^{\rm{H}} \xlongrightarrow[{N_{\rm{T}}} \to \infty]{ \rm{a.s.}}L{\rm{Tr}}\left\{ {{\bf{R}}_{\rm{B}}^{ - 1}} \right\}{\rm{ + }}A{\rm{Tr}}\left\{ {\overline {\bf{R}} } \right\}$.
 Similarly, we have ${\widehat {\bf{g}}_{\rm{E}}}\left( {{\bf{R}}_{\rm{E}}^{ - 1} \otimes {{\bf{I}}_L}} \right)\widehat {\bf{g}}_{\rm{E}}^{\rm{H}} \xlongrightarrow[{N_{\rm{T}}} \to \infty]{ \rm{a.s.}}L{\rm{Tr}}\left\{ {{\bf{R}}_{\rm{B}}^{ - 1}} \right\} + B{\rm{Tr}}\left\{ {\widetilde {\bf{R}}} \right\}$
  where $\widetilde {\bf{R}} = {\bf{R}}_{\rm{E}}^{ - \frac{1}{2}}{\bf{R}}_{\rm{B}}^{ - 1}{\bf{R}}_{\rm{E}}^{ - \frac{1}{2}}\otimes  \left\{{\left( {{\bf{F}}_{\rm{L}}^{\rm{T}}} \right)^{\rm{H}}}{\left( {{\bf{F}}_{\rm{L}}^{\rm{T}}{\bf{F}}_{\rm{L}}^{\rm{*}}} \right)^{ - 2}}{\bf{F}}_{\rm{L}}^{\rm{T}}\right\}$. Based on those equations, the proposition can be easily proved.
 \end{proof}
We use the asymptotic analysis  as a tool to provide tight approximations for finite $N_{\rm T}$~\cite{Hoydis}. As shown in Fig.~\ref{figure3}(b),  not very large antenna, i.e., $N_{\rm T}= 64$,  can  bring the precise decision. For massive MIMO systems, generally with  antennas of  128 or more, those asymptotic approximation  results are precise enough for our calculation.   Proposition 4  provides a mathematical support of the theoretical limit  that pilot identification  enhancement   can bring under randomly-imitating attack.

In what follows, we stimulate the performance of channel  estimation and identification  using proposed estimators. We consider uniform linear array (ULA) with spacing $d=\lambda/2$  and $N_{\rm T}\ge 8$.   All the spatial correlation matrices are generated  such that  the spatial correlation between any two antennas at each path can be given by $\rho \left( {{\theta _0},d} \right){\rm{ = }}\int_{ - \pi }^\pi  {{e^{\frac{{ - j2\pi d\sin \left( {\theta  - {\theta _0}} \right)}}{\lambda }}}} P\left( {\theta  - {\theta _0}} \right)d\theta$
where $\theta_0$ denotes the mean AoA and  $P$ denotes the channel power angle spectrum (PAS) that is  modeled by Truncated Gaussian  distribution~\cite{Cho,Xu_Joint2}. The mean AoA of Bob, Charlie and  Eva, respectively denoted by $\theta^{\rm B}_{0}$, $\theta^{\rm C}_{0}$ and $\theta^{\rm E}_{0}$,  are generated independently and distributed identically within $[-\pi,\pi]$. For the channel estimation part,  we consider  $N_{\rm Total}= 128$ subcarriers  are occupied by pilot tones.  We assume Bob and Charlie have the same transmission  power, i. e., $\rho _{\rm{B}}=\rho _{\rm{C}}=\rho$ and define  ${\rm{SNR}} = {{{P}} \mathord{\left/ {\vphantom {{{P}} {{\sigma ^2}}}} \right. \kern-\nulldelimiterspace} {{\sigma ^2}}}$.  We also define the user average MSE (UMSE) of  FS   and CIR estimation respectively  as
 ${{{\mathbb E}\left\{ {{{\left\| {{{\widehat{\bf{h}}}_{\rm{B}}} - {{\bf{h}}_{\rm{B}}}} \right\|}^2} + {{\left\| {{{\widehat{\bf{h}}}_{\rm{C}}} - {{\bf{h}}_{\rm{C}}}} \right\|}^2}} \right\}} \mathord{\left/
 {\vphantom {{E\left\{ {{{\left\| {{{\widehat{\bf{h}}}_{\rm{B}}} - {{\bf{h}}_{\rm{B}}}} \right\|}^2} + {{\left\| {{{\widehat{\bf{h}}}_{\rm{C}}} - {{\bf{h}}_{\rm{C}}}} \right\|}^2}} \right\}} {2N{N_{\rm{T}}}}}} \right.
 \kern-\nulldelimiterspace} {2N{N_{\rm{T}}}}}$ and ${{{\mathbb E}\left\{ {{{\left\| {{{\widehat{\bf{g}}}_{\rm{B}}} - {{\bf{g}}_{\rm{B}}}} \right\|}^2} + {{\left\| {{{\widehat{\bf{g}}}_{\rm{C}}} - {{\bf{g}}_{\rm{C}}}} \right\|}^2}} \right\}} \mathord{\left/
 {\vphantom {{E\left\{ {{{\left\| {{{\widehat{\bf{g}}}_{\rm{B}}} - {{\bf{g}}_{\rm{B}}}} \right\|}^2} + {{\left\| {{{\widehat{\bf{g}}}_{\rm{C}}} - {{\bf{g}}_{\rm{C}}}} \right\|}^2}} \right\}} {2N{N_{\rm{T}}}}}} \right.
 \kern-\nulldelimiterspace} {2N{N_{\rm{T}}}}}$.
For the identification simulations,  since Eva can flexibly choose to attack any nodes, we  define the identification error region (IER) as the set of  all the collections of $\left( {\theta _0^{\rm{B}} - \theta _0^{\rm{E}},\theta _0^{\rm{C}} - \theta _0^{\rm{E}}} \right)$ such that $ {\rm Pr_{{IEP}}}>0$ is satisfied for  Bob and/or Charlie. Correspondingly, the perfect identification region (PIR) is defined as the set making $ {\rm Pr_{{IEP}}}=0$.

 Fig.~\ref{figure3}(a) presents the  UMSE performance of CIR estimation versus SNR  and  different number $B$ of subcarrier blocks.  $L$ is configured to be 6 and Eva is assumed to be with same SNR as Bob and Charlie. As we can see,  traditional pilot spoofing attack causes a high-UMSE floor on  CIR estimation  for Bob and Charlie. However, the proposed mechanism   breaks down this floor and its UMSE  gradually decreases with the increase of transmit antennas.  Moreover, we find that  the UMSE under proposed estimators   approaches the  level under perfect MMSE with the increase of antennas. On the other hand, the case without quantization serves as an another performance benchmark.  It can be shown that the UMSE gradually decreases with  the increase of  $B$  and $B=5$   is enough  to guarantee Alice a desirable UMSE, like the one under no quantization.

Fig.~\ref{figure3}(b) shows the IEP versus the mean  AoA separation  under  $N_{\rm T}=64$. The simulation is averaged over 1000 runs, each of which performs 1000 channel average.   As we can see, IER is composed by the special  points for which  at least one of its axes has zero value.  It means that
the available PIR can be extensively achieved unless  any legitimate node has the same average  AoA as Eva.
Fig.~\ref{figure3}(c)  shows the UMSE performance of  FS and CIR estimation versus SNR  under various  power difference relative to Eva. $L$ is  configured to be 8 and ${B}$ is set to be 5. As we can see, the UMSE is not influenced by the  power of Eva, even with $30\rm dB$ larger  than Bob or Charlie, under both $N_{\rm T}=100$ and $N_{\rm T}=8$. The reason is that  the interference can be eliminated naturally from the received  signal space  when  the dimension of  signals observed is no more than the number of OFDM symbol time in use.

\section{Conclusions}
\label{Conclusions}
In this paper, we designed a CFBG based PA  mechanism for a two-user OFDM system to protect the channel estimation  over frequency-selective channels. In this scheme, the values of pilot tones were randomized to  avoid the pilot-spoofing attack but also cause a serious hybrid attack.  To resolve those problems in a unique framework, a  scheme combing  detection,  coding and  channel estimation  was devised to achieve secure PA with low SEP  and IEP   as well as   high-accuracy channel estimation.  Some interesting results were presented  to verify the robustness  of proposed scheme under hybrid attack modes.
%

\end{document}